\documentclass[letterpaper,11pt]{article} 

\newif\ifFULL
\FULLtrue

\usepackage[top=1in, bottom=1in, left=1in, right=1in]{geometry}

\usepackage{amsmath,amsthm, amssymb}
\usepackage{latexsym}
\usepackage{subcaption}
\usepackage{lipsum}
\usepackage{graphicx}
\usepackage{ifthen}
\usepackage{multicol}
\usepackage{algorithm,algorithmic} 
\usepackage[numbers]{natbib}
\usepackage{tikz}
\usepackage{pgf}
\usepackage{dsfont}
\usepackage[font={small}]{caption}
\usepackage{complexity}
\providecommand{\B}{\operatorname{B}}
\usepackage{color}
\usepackage[normalem]{ulem}

\usepackage{hyperref}
\hypersetup{colorlinks=true,linkcolor=blue,citecolor=blue,urlcolor=blue}

\newtheorem{lemma}{Lemma}
\newtheorem*{lemma*}{Lemma}
\newtheorem{theorem}{Theorem}
\newtheorem*{theorem*}{Theorem}
\newtheorem{definition}{Definition}
\newtheorem{corollary}{Corollary}
\newtheorem*{corollary*}{Corollary}
\newtheorem*{claim*}{Claim}
\numberwithin{theorem}{section}
\numberwithin{lemma}{section}
\newtheorem*{definition*}{Definition}

\usepackage{thmtools}
\usepackage{thm-restate}

\newlength\myindent
\newcommand\bindent{%
\begingroup
 \setlength{\itemindent}{\myindent}
 \addtolength{\algorithmicindent}{\myindent}
}
\newcommand\eindent{\endgroup}
\newcommand{\INPUT}{\item[{\bf Input:}]}
\renewcommand{\algorithmiccomment}[1]{\bgroup\hfill\footnotesize~#1\egroup}

\newcommand{\prob}[2][]{\textnormal{Pr}\ifthenelse{\not\equal{}{#1}}{_{#1}}{}\!\left(#2\right)}
\newcommand{\expect}[2][]{\text{\bf E}\ifthenelse{\not\equal{}{#1}}{_{#1}}{}\!\left[#2\right]}
\newcommand{\var}[2][]{\text{\bf Var}\ifthenelse{\not\equal{}{#1}}{_{#1}}{}\!\left[#2\right]}

\newcommand{\given}{\, : \,}

\newcommand{\ba}{\mathbf{a}}
\newcommand{\bx}{\mathbf{x}}
\newcommand{\by}{\mathbf{y}}
\newcommand{\bz}{\mathbf{z}}

\DeclareMathOperator{\argmax}{argmax}

\newcommand{\PC}{\mathcal{P}}   
\newcommand{\UC}{\mathcal{U}}   
\newcommand{\DC}{\mathcal{D}}   

\title{The Limitations of Optimization from Samples}
\author{Eric Balkanski\footnote{School of Engineering and Applied Sciences, Harvard University. ericbalkanski@g.harvard.edu.  This research was supported by Smith Family Graduate Science and Engineering Fellowship and by NSF grant CCF-1301976, CAREER CCF-1452961.} 
\qquad Aviad Rubinstein\footnote{Department of Electrical Engineering and Computer Sciences, UC Berkeley. aviad@eecs.berkeley.edu. 
This research was supported by Microsoft Research PhD Fellowship, NSF grant CCF1408635, and Templeton Foundation grant 3966, and done in part at the Simons Institute for the Theory of Computing.} 
\qquad Yaron Singer\footnote{School of Engineering and Applied Sciences, Harvard University. yaron@seas.harvard.edu.  This research was supported by NSF grant CCF-1301976, CAREER CCF-1452961.}}
\date{}
\begin{document}

\thispagestyle{empty}\maketitle\setcounter{page}{0}

\begin{abstract}
In this paper we consider the following question: can we optimize objective functions from the training data we use to learn them?  We formalize this question through a novel framework we call \emph{optimization from samples} (\texttt{OPS}). In \texttt{OPS}, we are given sampled values of a function drawn from some distribution and the objective is to optimize the function under some constraint.  

While there are interesting classes of functions that can be optimized from samples,  our main result is an impossibility.  We show that there are classes of functions which are statistically learnable and optimizable, but for which no reasonable approximation for optimization from samples is achievable.  In particular, our main result shows that there is no constant factor approximation for maximizing coverage functions under a cardinality constraint using polynomially-many samples drawn from any distribution.
%

We also show tight approximation guarantees for maximization under a cardinality constraint of several interesting classes of functions including unit-demand, additive, and general monotone submodular functions, as well as a constant factor approximation for monotone submodular functions with bounded curvature.
\end{abstract}
\newpage
\section{Introduction}

The traditional approach in optimization typically assumes there is an underlying model known to the algorithm designer, and the goal is to optimize an objective function defined through the model.  In a routing problem, for example, the model can be a weighted graph which encodes roads and their congestion, and the objective is to select a route that minimizes expected travel time from source to destination.  In influence maximization, we are given a weighted graph which models the likelihood of individuals forwarding information, and the objective is to select a subset of nodes to spread information and maximize the expected number of nodes that receive information~\cite{kempe2003maximizing}.  

In many applications like influence maximization or routing, we do not actually know the objective functions we wish to optimize since they depend on the behavior of the world generating the model.  In such cases, we gather information about the objective function from past observations and use that knowledge to optimize it.  A reasonable approach is to learn a surrogate function that approximates the function generating the data (e.g.~\cite{GLK10,DSGZ13,DGSS14,DLBS14a,DBS14b,NPS15}), and optimize the surrogate.  In routing, we may observe traffic, fit  weights to a graph that represents congestion times, and optimize for the shortest path on the weighted graph learned from data.  In influence maximization, we can observe information spreading in a social network, fit weights to a graph that encodes the influence model, and optimize for the $k$ most influential nodes.  But what are the guarantees we have?


One problem with optimizing a surrogate learned from data is that it may be inapproximable.  For a problem like influence maximization, for example, even if a surrogate $\widetilde{f}:2^N \to \mathbb{R}$ approximates a submodular influence function $f:2^N\to \mathbb{R}$ within a factor of $(1\pm \epsilon)$ for sub-constant $\epsilon>0$, in general there is no polynomial-time algorithm that can obtain a reasonable approximation to $\max_{S:|S|\leq k}\widetilde{f}(S)$ or $\max_{S:|S|\leq k}{f}(S)$~\cite{HS15}.   A different concern is that the function learned from data may be approximable (e.g. if the surrogate remains submodular), but its optima are very far from the optima of the function generating the data.  In influence maximization, even if the weights of the graph are learned within a factor of $(1\pm \epsilon)$ for sub-constant $\epsilon>0$ the optima of the surrogate may be a poor approximation to the true optimum~\cite{NPS15,HK16}.  The sensitivity of optimization to the nuances of the learning method therefore raises the following question:

\begin{center}
\emph{Can we actually optimize objective functions from the training data we use to learn them?}
\end{center}


%



\paragraph{Optimization from samples.}  In this paper we consider the following question: given an unknown objective function $f:2^N\to \mathbb{R}$ and samples $\{S_i,f(S_{i})\}_{i=1}^m$ where $S_i$ is drawn from some distribution $\mathcal{D}$ and $m \in \poly (|N|)$, is it possible to solve $\max_{S:|S|\leq k}f(S)$? More formally:

\begin{definition*}
A class of functions $\mathcal{F}:2^N\to \mathbb{R}$ is \textbf{\emph{$\alpha$-optimizable in $\mathcal{M}$ from samples over  distribution $\mathcal{D}$}}  if there exists a (not necessarily polynomial time) algorithm whose input is a set of samples $\{S_{i},f(S_i)\}_{i=1}^{m}$, where $f \in \mathcal{F}$ and $S_i$ is drawn i.i.d. from $\mathcal{D}$, and returns $S \in \mathcal{M}$ s.t.:
$$\Pr_{S_1,\dots,S_m \sim \mathcal{D}} 
\left[ \expect{f(S)} \geq \alpha \cdot \max_{T\in  \mathcal{M}}f(T) \right] \geq 1 - \delta,$$
where the expectation is over the decisions of the algorithm, $m \in \poly (|N|)$, $\delta \in [0 , 1)$ is a constant.  
\end{definition*}

An algorithm with the above guarantees is an $\alpha$-\texttt{OPS} algorithm. In this paper we focus on the simplest constraint, where $\mathcal{M} = \{S\subseteq N : |S|\leq k\}$ is a cardinality constraint.  For a class of functions $\mathcal{F}$ we say that optimization from samples is \emph{possible} when there exists some constant $\alpha \in (0,1]$ and any distribution $\mathcal{D}$ s.t. $\mathcal{F}$ is $\alpha$-optimizable from samples over $\mathcal{D}$ in $\mathcal{M} = \{S:|S|\leq k\}$.    

Before discussing what is achievable in this framework, the following points are worth noting:
\begin{itemize}
\item Optimization from samples is defined per distribution.  Note that if we demand optimization from samples to hold on all distributions, then trivially no function would be optimizable from samples (e.g. for the distribution which always returns the empty set);  

\item Optimization from samples seeks to approximate the global optimum.  
In learning, we evaluate a hypothesis on the same distribution we use to train it since it enables making a prediction about events that are similar to those observed.  For optimization it is trivial to be competitive against a sample by simply selecting the feasible solution with maximal value from the set of samples observed.  Since an optimization algorithm has the power to select any solution, the hope is that polynomially many samples contain enough information for optimizing the function.  In influence maximization, for example, we are interested in selecting a set of influencers, even if we did not observe a set of highly influential individuals that initiate a cascade together.
\end{itemize}

As we later show,  there are interesting classes of functions and distributions that indeed allow us to approximate the global optimum well, in polynomial-time using polynomially many samples.  The question is therefore not whether optimization from samples is possible, but rather  which function classes are optimizable from samples.  









%

\subsection{Optimizability and learnability}
Optimization from samples is particularly interesting when functions are \emph{learnable} and \emph{optimizable}.




\begin{itemize}
\item \textbf{Optimizability.}  We are interested in functions $f:2^N \to \mathbb{R}$ and  constraint $\mathcal{M}$ such that given access to a \emph{value oracle} (given $S$ the oracle returns $f(S)$), there exists a constant factor approximation algorithm for $\max_{S\in \mathcal{M}}f(S)$.   For this purpose, monotone submodular functions are a convenient class to work with, where the canonical problem is $\max_{|S|\leq k}f(S)$.  It is well known that there is a $1-1/e$ approximation algorithm for this problem~\cite{NWF78} and that this is tight using polynomially many value queries~\cite{feige1998threshold}.  Influence maximization is an example of maximizing a monotone submodular function under a cardinality constraint~\cite{kempe2003maximizing}.  
%
\item \textbf{PMAC-learnability.} The standard framework in the literature for learning set functions is   \emph{Probably Mostly Approximately Correct} ($\alpha$-\texttt{PMAC}) learnability due to Balcan and Harvey~\cite{BH11-PMAC}.  This framework nicely generalizes Valiant's notion of \emph{Probably Approximately Correct} (\texttt{PAC}) learnability~\cite{Valiant84-PAC}.  Informally, PMAC-learnability guarantees that after observing polynomially many samples of sets and their function values, one can construct a surrogate function that is likely to, $\alpha$-approximately, mimic 
the behavior of the function observed from the samples (see Appendix~\ref{s:learningmodels} for formal definitions).  Since the seminal paper of  Balcan and Harvey, there has been a great deal of work on learnability of submodular functions~\cite{FK14-coverage, BCIW12-valuations,  BDFKNR12-sketches, FV13-submodular_juntas, feldman2015tight, Balcan15-AAMAS}. 
\end{itemize}
 


Are functions that are learnable and optimizable also optimizable from samples?

\subsection{Main result}
Our main result is an impossibility.  We show that there is an interesting class of functions that is \texttt{PMAC}-learnable and optimizable but not optimizable from samples.  This class is coverage functions.

\begin{definition*}  A function is called \textbf{\emph{coverage}} if there exists a family of sets $T_{1},\ldots, T_{n}$ that covers subsets of a universe $U$  with weights $w(a_j)$ for $a_j \in U$ such that for all $S$,  $f(S) = \sum_{a_j \in \cup_{i \in S}T_i}w(a_j)$.  A coverage function is \emph{\textbf{polynomial-sized}} if the universe is of polynomial size in $n$.  Influence maximization is a generalization of maximizing coverage functions under a cardinality constraint. 
\end{definition*}

Coverage functions are a canonical example of monotone submodular functions and are hence optimizable.  In terms of learnability, for any constant $\epsilon>0$, coverage functions are $(1-\epsilon)$-\texttt{PMAC} learnable over any distribution~\cite{BDFKNR12-sketches}, unlike monotone submodular functions which are generally not \texttt{PMAC} learnable~\cite{BH11-PMAC}.  Somewhat surprisingly, coverage functions are not optimizable from samples.



\begin{theorem*}
No algorithm can obtain an approximation better than $2^{-\Omega(\sqrt{\log n})}$ for maximizing a polynomial-sized coverage function under a cardinality constraint, using polynomially many samples drawn from any distribution.
\end{theorem*}

Coverage functions are heavily used in machine learning~\cite{swaminathan2009essential, yue2008predicting, guestrin2005near, krause2007near, antonellis2012dynamic, lin2011class, takamura2009text}, data-mining~\cite{chierichetti2010max, du2014learning, saha2009maximum, singer2012win, dasgupta2007discoverability, gomez2010inferring}, mechanism design~\cite{dobzinski2006improved, lehmann2001combinatorial, dughmi2015limitations, dughmi2011convex,buchfuhrer2010computation,dughmi2011truthful}, privacy~\cite{gupta2013privately, FK14-coverage}, as well as influence maximization~\cite{kempe2003maximizing, seeman2013adaptive, borgs2014maximizing}.  In many of these applications, the functions are learned from data and the goal is to optimize the function under a cardinality constraint.  In addition to learnability and optimizability, coverage functions have many other desirable properties (see Section~\ref{s:discussion}).  One important fact is that they are \emph{parametric}: if the sets $T_1,\ldots,T_n$ are known, then the coverage function is completely defined by the weights $\{w(a) \ : \ a \in U \}$.  Our impossibility result holds even in the case where the sets $T_1,\ldots,T_n$ are known. 

\paragraph{Technical overview.}  In the value query model, information theoretic impossibility results use functions defined over a partition of the ground set \cite{MSV08, Vondrak13,EneVW13}. The hardness then arises from hiding all the information about the partition from the algorithm. Although the constructions in the \texttt{OPS} model also rely on a partition, the techniques are different since the impossibility is quasi-polynomial and not constant. In particular, the algorithm may learn the entire partition, and the hardness arises from hiding which parts of the partition are ``good" or ``bad". We begin by describing a framework which reduces the problem of showing hardness results to constructing good and bad functions which satisfy certain properties. The desired good and bad functions must have equal value on small sets of equal sizes and a large gap in value on large sets. Interestingly, a main technical difficulty is to simultaneously satisfy these two simple properties, which we do with novel techniques for constructing coverage functions. Another technical part is the use of tools from pseudorandomness  to obtain coverage functions of polynomial size.

\subsection{Algorithms for OPS}  

There are classes of functions and distributions for which optimization from samples is possible.  Most of the algorithms use a simple technique that consists of estimating the expected marginal contribution of an element to a random sample.  For general submodular functions, we show an essentially tight bound using a non-trivial analysis of  an algorithm that uses such estimates.

\begin{theorem*}
There exists an $\tilde{\Omega}(n^{-1/4})$-\texttt{OPS} algorithm over a distribution $\DC$ for monotone submodular functions.  Furthermore, this approximation ratio is essentially tight.  
\end{theorem*}

For unit-demand and  additive functions, we give near-optimal  optimization from samples results.  The result for unit-demand is particularly interesting as it shows one can easily optimize a function from samples even when recovering it is impossible (see Section \ref{s:recoverability}). For monotone submodular functions with curvature $c$, we obtain a $((1-c)^2-o(1))$-\texttt{OPS} algorithm.

\subsection{Paper Organization}

We begin with the hardness result in Section~\ref{s:impossibility}. The \texttt{OPS} algorithms are presented in Section~\ref{s:algorithms}.  We discuss the notion of recoverability in Section \ref{s:recoverability} and additional related work in Section~\ref{s:addrelatedwork}. The proofs are deferred to the appendix.

\section{Impossibility of Optimization from Samples}
\label{s:impossibility}

We show that optimization from samples is in general impossible, over any distribution $\mathcal{D}$, even when the function is learnable and optimizable. Specifically, we show that there exists no constant $\alpha$ and distribution $\mathcal{D}$ such that coverage functions are $\alpha$-optimizable from samples, even though they are $(1 - \epsilon)$-\texttt{PMAC} learnable over any distribution $\mathcal{D}$ and can be maximized under a cardinality constraint within a factor of $1 - 1/e$. In Section~\ref{s:framework}, we construct a framework which reduces the problem of proving information theoretic lower bounds to constructing functions that satisfy certain properties. We then construct coverage functions that satisfy these properties in Section~\ref{s:lbcoverage}.

\subsection{A Framework for OPS Hardness}
\label{s:framework}

The framework we introduce partitions the ground set of elements into \emph{good, bad,} and \emph{masking} elements. We derive two conditions on the values of these elements so that samples do not contain enough information to distinguish good and bad elements with high probability. We then give two additional conditions so that if an algorithm cannot distinguish good and bad elements, the solution returned by this algorithm has low value compared to the optimal set consisting of the good elements. We begin by defining the partition.

\begin{definition*}
The collection of partitions $\mathcal{P}$ contains all partitions $P$ of the ground set $N$ in $r$ parts $T_1, \ldots, T_r$ of $k$  elements and a part $M$ of remaining $n - rk$  elements, where $n = |N|$.
\end{definition*}

The elements in $T_i$ are called the \emph{good} elements, for some $i \in [r]$. The \emph{bad} and \emph{masking} elements are the elements in $T_{-i} := \cup_{j = 1, j \neq i}^{r}T_j$ and $M$ respectively. Next, we define a class of functions $ \mathcal{F}(g,b,m,m^+)$ such that $f \in \mathcal{F}(g,b,m,m^+)$ is defined in terms of good, bad, and masking  functions $g$, $b$, and $m^+$, and a masking fraction $m \in [0,1]$.\footnote{The notation $m^+$ refers to the role of this function, which is to maintain monotonicity of masking elements. These four functions are assumed to be normalized such that $g(\emptyset) = b(\emptyset) = m(\emptyset) = m^+(\emptyset) = 0$.} 

\begin{definition*}
Given functions $g,b,m,m^{+}$, the class of  functions $\mathcal{F}(g,b,m,m^+)$ contains  functions $f^{P,i}$,  where $P \in \mathcal{P}$ and $i \in [r]$, defined as
 $$f^{P,i}(S) := (1 - m(S \cap M)) \Big(g(S \cap T_i) + b(S \cap T_{-i})\Big) + m^+(S \cap  M).$$
\end{definition*}
 We use probabilistic arguments over the partition $P \in \mathcal{P}$ and the integer $i \in [r]$ chosen uniformly at random to show that for any distribution $\mathcal{D}$ and any algorithm, there exists a function in $\mathcal{F}(g,b,m,m^+)$  that the algorithm optimizes poorly given samples from $\DC$.   The functions $g,b,m,m^{+}$ have desired properties that are parametrized below. At a high level, the identical on small samples and masking on large samples properties imply that the samples do not contain enough information to learn $i$, i.e. distinguish good and bad elements, even though the partition $P$ can be learned. The gap and curvature property imply that if an algorithm cannot distinguish good and bad elements, then the algorithm performs poorly.

\begin{definition*}
The class of functions $\mathcal{F}(g,b,m,m^{+})$ has an \textbf{$(\alpha, \beta)$-gap} if the following conditions are satisfied for some $t$, where $\mathcal{U}(\PC)$ is the uniform distribution over $\PC$.
\begin{enumerate}
\item  \textbf{Identical on small samples.} For a fixed $S:|S| \leq t$, with probability $1 - n^{-\omega(1)}$ over partition $P \sim \mathcal{U}(\PC)$, $g(S \cap T_{i}) + b(S \cap T_{- i})$ is independent of $i$;

\item \textbf{Masking on large samples.} For a fixed $S :|S| \geq  t$, with probability $1 - n^{-\omega(1)}$ over partition $P \sim \mathcal{U}(\PC)$, the masking fraction is $m(S \cap M) = 1$;  
\item \textbf{$\alpha$-Gap.} Let $S: |S| =k$, then $g(S) \geq \max\{\alpha  \cdot b(S),\alpha \cdot m^+(S)\}$;
\item \textbf{$\beta$-Curvature.} Let $S_1:|S_1| = k$ and $S_2: |S_2| = k / r$, then
$g(S_1) \geq (1 - \beta) \cdot r \cdot g(S_2)$. 
\end{enumerate}
\end{definition*}

The following lemma reduces the problem of showing an impossibility result to constructing $g, b,m,$ and  $m^+$ which satisfy the above properties.

\begin{restatable}{rThm}{tgap}
\label{t:gap}
Assume the functions $g,b,m,m^{+}$ have an $(\alpha, \beta)$-gap, then  $\mathcal{F}(g,b,m,m^{+})$  is not $2\max(  1 /(r ( 1- \beta)), 2/\alpha)$-optimizable from samples over any distribution $\mathcal{D}$.
\end{restatable}

Consider a distribution $\DC$. The proof of this result consists of three parts. 
\begin{enumerate}
\item Fix a set $S$. With probability  $1 - n^{-\omega(1)}$ over  $P \sim \mathcal{U}(\PC)$, $f^{P,i}(S)$ is independent of $i$, by the identical on small samples and masking on large samples properties.
\item There exists a partition $P \in \mathcal{\PC}$ such that with probability $1- n^{-\omega(1)}$ over polynomially many samples $\mathcal{S}$ drawn from $\DC$, $f^{P,i}(S)$ is independent of $i$ for all $S \in \mathcal{S}$.  Thus,  given samples $\{(S_j,  f^{P,i}(S_j))\}_j$ for such a partition $P$, the decisions of the algorithm are independent of $i$. 
\item There exists  $f^{P,i}$ such that the algorithm does not obtain a $2\max(  1 /(r ( 1- \beta)), 2/\alpha)$ approximation for  $f^{P,i}$ with samples from $\DC$. This holds by a second probabilistic argument, this time over $i \in \UC([r])$, and by the gap and curvature properties.
\end{enumerate}

\subsection{OPS Hardness of Coverage Functions}
\label{s:lbcoverage}

We use this framework to show that there exists no constant $\alpha$ and distribution $\mathcal{D}$ such that coverage functions are $\alpha$-optimizable from samples over $\mathcal{D}$. We first state a  definition of coverage functions that is equivalent to the traditional definition and that is used through this section.
\begin{definition}
A function $f : 2^N \rightarrow \mathbb{R}$ is coverage if there exists a bipartite graph $G = (N \cup \{a_j\}_j, E)$ between elements $N$ and children $\{a_j\}_j$ with weights $w(a_j)$ such that the value of a set $S$ is the sum of the weights of the children covered by $S$, i.e., for all $S \subseteq N$, $f(S) = \sum_{a_j: \exists e a_j \in E, e \in S} w(a_j)$. A coverage function is polynomial-sized if the number of children is polynomial in $n = |N|$.
\end{definition}

The construction of good and bad coverage functions $g$ and $b$ that combine the identical on small samples property  and a large $\alpha$-gap on large sets as needed by the framework is a main technical challenge. The bad function $b$ needs to increase slowly (or not at all) for large sets to obtain a large $\alpha$-gap, which requires a non-trivial overlap in the children covered by bad elements (this is related to coverage functions being {\em second-order supermodular} \cite{KMZ15-online_SWM}). The overlap in children covered by good elements then must be similar  (identical on small samples) while the good function still needs to grow quickly for large sets (large gap), as illustrated in Figure~\ref{f:goodbad}. We consider the cardinality constraint $k = n^{2/5 - \epsilon}$ and a number of parts $r = n^{1/5 - \epsilon}$.   At a high level, the proof follows three main steps.

 
 \begin{figure}
\centering
 \includegraphics[width=.33\linewidth]{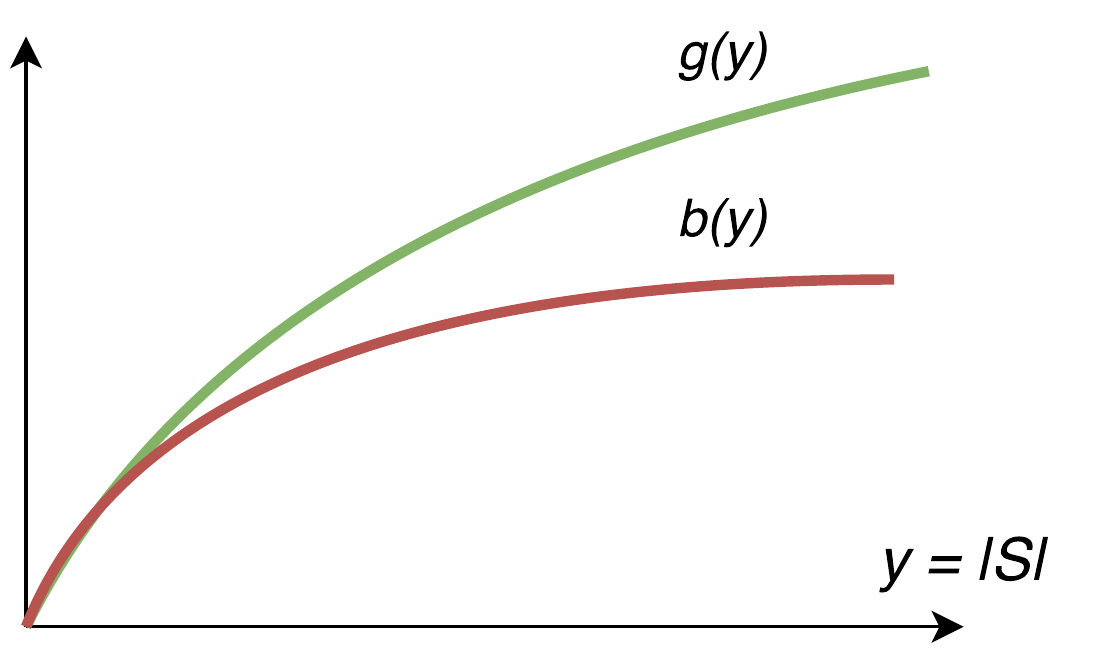}
\caption{Sketches of the desired $g(S)$ and $b(S)$ in the simple case where they only depend on $y = |S|$.}
\label{f:goodbad}
\end{figure}

 
\begin{enumerate}
\item \textbf{Constructing the good and bad functions.}  In  Section~\ref{s:exp-construction}, we construct the good and bad functions whose values are identical on small samples for $t = n^{3/5 + \epsilon}$, have gap $\alpha = n^{1/5-\epsilon}$, and curvature $\beta = o(1)$.  These good and bad functions are affine combinations of primitives $\{C_{p}\}_{p \in \mathbb{N}}$ which are coverage functions with desirable properties;
  
\item \textbf{Constructing the masking function.} In Section~\ref{s:masking}, we construct $m$ and $m^+$ that are masking on large samples for $t = n^{3/5 + \epsilon}$ and that have a gap  $\alpha = n^{1/5}$. In this construction, masking elements cover the children from functions $g$ and $b$ such that $t$ masking elements cover all the children, but $k$ masking elements only covers an $n^{-1/5}$ fraction of them.
\item \textbf{From exponential to polynomial-sized coverage functions.}  Lastly in Section~\ref{s:succinct} we prove the hardness result for polynomial-sized coverage functions. This construction relies on constructions of $\ell$-wise independent variables to reduce the number of children.
\end{enumerate}

\subsubsection{Constructing the Good and the Bad Coverage Functions} \label{s:exp-construction}
In this section we describe the construction of good and bad functions that are identical on small samples for $t = n^{3/5 + \epsilon}$, with a gap $\alpha = n^{1/5-\epsilon}$ and curvature $\beta = o(1)$.  To do so, we introduce a class of primitive functions $C_{p}$, through which we express the good and bad functions.  For symmetric functions $h$ (i.e. whose value only depends on the size of the set), we abuse notation and simply write $h(y)$ instead of $h(S)$ for a set $S$ of size $y$.

%
%

\paragraph{The construction.}  We begin by describing the primitives we use for the good and bad functions.  These primitives are the family $\{C_{p}\}_{p\in \mathbb{N}}$, which are symmetric, and defined as: 
$$C_{p}(y) = p \cdot \left( 1 - (1 - 1/p)^y \right).$$
%


 These are \emph{coverage} functions defined over an \emph{exponential} number of children.

\begin{restatable}{rClaim}{cct}
\label{c:c1t}
Consider the coverage function over ground set $N$ where for each set $S$, there is a child $a_S$ that is covered by exactly $S$, and child $a_S$ has weight $w(a_S) = p \cdot \prob{S \sim \B(N, 1/p)}$ where the binomial distribution $\B (N, 1/p)$ picks each element in $N$ independently with probability $1/p$, then this coverage function is $C_{p}$.
\end{restatable}

For a given $\ell \in [n]$, we construct $g$ and $b$ as affine combinations of $\ell$ coverage functions $C_{p_j}(y)$ weighted by variables $x_j$ for $j \in [\ell]$:

\begin{itemize}
\item \textbf{The good function} is defined as: 
$$g(y) := y + \sum_{j \given x_j < 0} (-x_j) C_{p_j}(y)$$
\item \textbf{The bad function}  is defined as $b(S) = \sum_{j=1, j \neq i}^r b'(S \cap T_j)$, with
$$b'(y) := \sum_{j \given x_j > 0} x_j C_{p_j}(y)$$ 
\end{itemize}

\paragraph{Overview of the analysis of the good and bad functions.}  Observe that if $g(y) = b'(y)$ for all $y \leq \ell$ for some sufficiently large $\ell$, then we obtain the identical on small samples property. The main idea is to express these $\ell$ constraints  as a system of linear equations $A \bx = \by$ where $A_{ij} := C_{p_j}(i)$ and $y_j := j$, with $i,j \in [\ell]$.  We prove that this matrix has two crucial properties:  
\begin{enumerate}
\item \textbf{$A$ is invertible.}  In Lemma~\ref{lem:invertible} we show that there exists $\{p_j\}_{j=1}^{\ell}$ such that the matrix $A$ is invertible by interpreting its entries defined by $C_{p_j}$ as  non-zero polynomials of degree $\ell$.  This implies that the system of linear equations $A\cdot \bx = \by$ can be solved and that there exists a coefficient $\bx^*$ needed for our construction of the good and the bad functions; 

\item \textbf{$||\bx^{\star}||_{\infty}$ is bounded.}  In Lemma~\ref{l:boundCoeff} we use Cramer's rule and Hadamard's inequality to prove that the entries of $\bx^{\star}$ are bounded.   This implies that the linear term $y$ in $g(y)$ dominates $x^{\star}_j \cdot C_{p_j}(y)$ for large $y$ and all $j$.  This then allows us to prove the curvature and gap properties.
\end{enumerate}

These properties of $A$ imply the desired properties of $g$ and $b$ for $\ell = \log \log n$.

\begin{restatable}{rLem}{labconstruction}
\label{l:abconstruction}
For every constant $\epsilon>0$, there exists coverage functions $g,b$ such that the identical on small samples property holds for $t = n^{3/5 + \epsilon}$, with gap $\alpha = n^{1/5-\epsilon}$ and  curvature $\beta = o(1)$.
\end{restatable}

Lemma \ref{l:issp} shows the identical on small samples property. It uses Lemma~\ref{lem:skewed-concentration} which shows that if $|S| \leq n^{3/5 + \epsilon}$, then with probability $1 - n^{-\omega(1)}$, $|S \cap T_j| \leq \log \log n$ for all $j$. The property then follows from the system of linear equations.  The gap  and curvature properties are proven in Lemmas~\ref{l:gap} and \ref{l:curv} using the fact that the term $y$ in $g$ dominates the other terms in $g$ and $b$.

\subsubsection{Constructing the Masking Function}
\label{s:masking}

Masking elements allow the indistinguishability of good and bad elements from large samples.

\paragraph{The masking elements.}  The construction of the coverage functions $g$ and $b$ defined in the previous section is generalized  so that we can add masking elements with desirable properties.  For each child $a_i$ in the coverage function defined by $g + b$, we divide $a_i$ into $n^{3/5}$ children $a_{i,1}, \ldots, a_{i,n^{3/5}}$ with equal weights $w(a_{i,j}) = \frac{w(a_i)}{n^{3/5}}$ for all $j$.  Each element covering $a_i$ according to $g$ and $b$ now covers children $a_{i,1}, \ldots, a_{i,n^{3/5}}$.  Note that the value of $g(S)$ and $b(S)$ remains unchanged with this new construction and thus, the previous analysis still holds.  Each masking elements in $ M$ is  defined by drawing $j \sim \UC{[n^{3/5}]}$   and having this element cover children $a_{i,j}$ for all $i$.   

The masking function $m^+(S)$ is the total weight covered by masking elements $S$ and the masking fraction $m(S)$ is the fraction of $j \in [n^{3/5}]$ such that $j$ is drawn for at least one  element in $S$. 

\paragraph{Masking properties.}  Masking elements cover children that are already covered by good or bad elements. A large number of masking elements mask the good and bad elements, which implies that good and bad elements are indistinguishable.  

\begin{itemize}
\item In Lemma~\ref{l:masking} we prove that the masking property holds for $t = n^{3/5 + \epsilon}$.  
\item We show a gap $\alpha = n^{1/5}$ in Lemma~\ref{l:gapmasking}.  For any $S:|S|\leq k$, we have $g(S)\geq n^{1/5}\cdot m^+(S)$.
\end{itemize}

\paragraph{An impossibility result for exponential size coverage functions.} We have the four properties for a $(n^{1/5 - \epsilon}, o(1))$-gap. The functions $f^{P,i}$ are coverage since $g,b,m^+$ are coverage and $m^+$ is the fraction of overlap between children from $g,b,$ and $m^+$

\begin{restatable}{rClaim}{cexpo}
\label{c:expo}
Coverage functions are not $n^{-1/5 + \epsilon}$-optimizable from samples over any distribution $\mathcal{D}$, for any constant $\epsilon > 0$.
\end{restatable}

\subsubsection{From Exponential to Polynomial Size Coverage Functions} \label{s:succinct}
The construction above relies on the primitives $C_{p}$ which are defined with exponentially many  children.  In this section we modify the construction to use primitives $c_{p}$ which are coverage with \emph{polynomially} many children.  The function class $\mathcal{F}(g,b,m,m^+)$ obtained are then coverage functions with polynomially many children. The functions $c_{p}$ we construct satisfy $c_{p}(y) = C_{p}(y)$ for all $y \leq \ell$, and thus the matrix $A$ for polynomial size coverage functions is identical to the general case.  We lower the cardinality constraint to $k = 2^{\sqrt{\log n}} = |T_j|$ so that the functions $c_{p}(S \cap T_j)$ need to be defined over only $2^{\sqrt{\log n}}$ elements. We also lower the number of parts to $r = 2^{\sqrt{\log n}/2} $.

\paragraph{Maintaining symmetry via $\ell$-wise independence.}  The technical challenge in defining a coverage function with polynomially many children is in maintaining the symmetry of non-trivial size sets.  To do so, we construct coverage functions $\{\zeta^z\}_{z\in [k]}$ for which the elements that cover a random child are approximately $\ell$-wise independent.  The next lemma reduces the problem to that of constructing coverage functions $\zeta^z$ that satisfy certain properties.


\begin{restatable}{rLem}{lzeta}
\label{l:zeta}
Assume there exist symmetric (up to sets of size $\ell$) coverage functions $\zeta^z$ with $\poly(n)$ children that are each covered by $z\in [k]$ parents. Then, there exists coverage functions $c_{p}$ with $\poly(n)$ children that satisfy $c_{p}(S) = C_{p}(y)$ for all $S$ such that $|S|= y \leq \ell$,  and $c_{p}(k) = C_{p}(k)$.
\end{restatable}
The proof is constructive. We obtain $c_{p}$ by replacing, for all $z \in [k]$,  all children in $C_{p}$ that are covered by $z$ elements with children from $\zeta^z$ with weights normalized that sum up $C_{p}(k)$.   Next, we construct such $\zeta^z$. Assume without loss that $k$ is prime (otherwise pick some prime close to $k$).  Given $\ba \in [k]^{\ell}$, and $x \in [z]$,  let
$$h_{\ba}(x) := \sum_{i \in [\ell]} a_i x^i \mod k$$ 
The children in $\zeta^z$ are $U = \{ \ba \in [k]^{\ell} \given h_{\ba}(x_1) \neq h_{\ba}(x_2) \text{ for all distinct } x_1, x_2 \in [z]\}$. The $k$ elements are $\{j \given 0 \leq j < k\}$. Child $\ba$ is covered by elements $\{h_{\ba}(x) \given x \in [z]\}$. Note that $|U| \leq k^{\ell} = 2^{\ell \sqrt{\log n}}$ and we pick $\ell = \log \log n$ as previously. The next lemma shows that we obtain the desired properties for $\zeta^z$.
\begin{restatable}{rLem}{lsymm}
\label{l:symm}
The coverage function $\zeta^z$ is symmetric for all sets of size at most $\ell$.
\end{restatable}
At a high level, the proof uses Lemma~\ref{l:indep} which shows that the parents of a random child $\ba$ are approximately $\ell$-wise independent. This follows  from $h_{\ba}(x)$ being a polynomial of degree $\ell - 1$, a standard construction for $\ell$-wise independent random variables. Then, using inclusion-exclusion over subsets $T$ of a set $S$ of size at most $\ell$, the probability that $T$ is the parents of a child $\ba$ only depends on $|T|$ by Lemma~\ref{l:indep}. Thus,  $\zeta^z(S)$ only depends on $|S|$. We are now ready to show the properties for $g,b,m,m^+$ with polynomially many children,

\begin{restatable}{rLem}{lpoly}
\label{l:poly}
There exists polynomial-sized coverage functions $g, b, m,$ and $m^+$ that satisfy an $(\alpha = 2^{\Omega(\sqrt{\log n})}, \beta= o(1))$-gap with $t = n^{3/5 + \epsilon}$.
\end{restatable}
We construct $g,b,m,m^+$ as in the general case but in terms of primitives $c_{p}$ instead of $C_{p}$. By Lemmas~\ref{l:zeta} and \ref{l:symm}, we obtain the same matrix $A$ and coefficients $\bx^{\star}$ as in the general case, so the identical on small samples property holds. The masking on large samples and curvature property hold almost identically as previously. Finally, since $k$ is reduced, the gap $\alpha$ is reduced to $2^{\Omega(\sqrt{\log n})}$.
\paragraph{OPS Hardness for Coverage Functions.} We get our main result by combining Theorem~\ref{t:gap} with  this $(\alpha = 2^{\Omega(\sqrt{\log n})}, \beta= o(1))$-gap.

\begin{restatable}{rThm}{tcover}
\label{t:cover}
For every constant $\epsilon > 0$, coverage functions are not $n^{-1/5 + \epsilon}$-optimizable from samples over any distribution $\DC$. In addition, polynomial-sized coverage functions are not $2^{-\Omega(\sqrt{\log n})}$-optimizable from samples over any distribution $\DC$.
\end{restatable}

\section{Algorithms for OPS}\label{s:algorithms} 

In this section we describe \texttt{OPS}-algorithms. Our algorithmic approach is not to learn a surrogate function and to then optimize this surrogate function. Instead, the algorithms estimate the expected marginal contribution of elements to a random sample directly from the samples (Section~\ref{s:tool}) and solve the optimization problem using these estimates. The marginal contribution of an element $e$ to a set $S$ if $f_S(e) := f(S \cup\{e\}) - f(S)$. If these marginal contributions are decreasing, i.e., $f_S(e) \geq f_T(e)$ for all $e \in N$ and $S \subseteq T \subseteq N$, then $f$ is submodular. If they are positive, i.e., $f_S(e) \geq 0$ for all $e \in N$ and $S \subseteq N$, then $f$ is monotone.

 This simple idea turns out to be quite powerful; we   use these estimate to develop an $\tilde{\Omega}(n^{-1/4})$ \texttt{OPS}-algorithm for monotone submodular functions in Section~\ref{s:submodular}. This approximation is essentially tight with a hardness result for general submodular functions shown in Appendix~\ref{s:lbsubmodular} that uses the framework from the previous section.  In Section~\ref{s:curv} we show that when samples are generated from a product distribution, there are interesting classes of functions that are amenable to optimization from samples.  
%
%
%

\subsection{OPS via Estimates of Expected Marginal Contributions}
\label{s:tool}
A simple case in which the expected marginal contribution $\expect[S \sim \DC|e_i \not \in S]{f_S(e_i)}$ of an element $e_i$ to a random set $S \sim \DC$ can be estimated arbitrarily well is that of product distributions.  We now show a simple algorithm we call \textsc{EEMC} which estimates the expected marginal contribution of an element when the distribution $\DC$ is a product distribution. This estimate is simply the difference between the average value of a sample containing $e_i$ and the average value of a sample not containing $e_i$.
 \begin{algorithm}[H]
\caption{\textsc{EEMC} Estimates the Expected Marginal Contribution $\expect[S \sim \DC|e_i \not \in S]{f_S(e)}$.}
\begin{algorithmic}
    	\INPUT $\mathcal{S} = \{S_j \given (S_j, f(S_j)) \text{ is a sample} )\}$
    	\FOR {$i \in [n]$ }
     \STATE $\mathcal{S}_i \leftarrow \{S \given S \in  \mathcal{S}, e_i \in S\}$
     \STATE $\mathcal{S}_{-i} \leftarrow \{S \given S \in  \mathcal{S}, e_i \not \in S\}$
    	\STATE $\hat{v}_i \leftarrow  \frac{1}{|\mathcal{S}_i|} \sum_{S \in \mathcal{S}_i} f(S) - \frac{1}{|\mathcal{S}_{-i}|} \sum_{S \in \mathcal{S}_{-i}} f(S)$
    	\ENDFOR
    	\RETURN $(\hat{v}_1, \ldots, \hat{v}_n)$
  \end{algorithmic}
\end{algorithm}
 
 \begin{restatable}{rLem}{lest}
\label{l:est}
Let $\mathcal{D}$ be a product distribution with bounded marginals.\footnote{The marginals are bounded if for all $e$, $e \in S\sim \DC$ and $e \not \in S\sim \DC$  w.p. at least $1/ \poly(n)$ and at most $1 - 1/\poly(n)$.} Then, with probability at least $1 - O(e^{-n})$, the estimations $\hat{v}_i$ are $\epsilon$ accurate, for any $\epsilon \geq f(N) / \poly(n)$ and for all $e_i$,  i.e., 
$$ | \hat{v}_i - \expect[S \sim \DC|e_i \not \in S]{f_S(e_i)} | \leq \epsilon.$$
\end{restatable}

The proof consists of the following two steps. First note that
$$\expect[S \sim \DC|e_i \not \in S]{f_S(e_i)} = \expect[S \sim \DC|e_i \not \in S]{f(S \cup e_i)} - \expect[S \sim \DC|e_i \not \in S]{f(S)} = \expect[S \sim \DC|e_i  \in S]{f(S)} - \expect[S \sim \DC|e_i \not \in S]{f(S)}$$where the second equality is since $\DC$ is a product distribution.
 Then, from standard concentration bounds, the average value  $(\sum_{S \in \mathcal{S}_i} f(S)) / |\mathcal{S}_i|$ of a set containing $e_i$ estimates $\expect[S \sim \DC|e_i  \in S]{f(S)}$ well. Similarly, $(\sum_{S \in \mathcal{S}_{-i}} f(S)) / |\mathcal{S}_{-i}|$ estimates $\expect[S \sim \DC|e_i \not \in S]{f(S)}$.

 \subsection{A Tight Approximation for Submodular Functions}\label{s:submodular}
We develop an $\tilde{\Omega}(n^{-1/4})$-\texttt{OPS} algorithm over $\DC$ for monotone submodular functions, for some distribution $\DC$.  This bound is essentially tight since submodular functions are not $n^{-1/4 + \epsilon}$-optimizable from samples over \emph{any} distribution (Appendix~\ref{s:lbsubmodular}).    We first describe the distribution for which the approximation holds. Then we describe the algorithm, which builds upon estimates of expected marginal contributions.


\paragraph{The distribution.} Let $\mathcal{D}_i$ be the uniform distribution over all sets of size $i$. Define the distribution $\mathcal{D}^{sub}$ to be the distribution which draws from $\mathcal{D}_k$, $\mathcal{D}_{\sqrt{n}}$, and $\mathcal{D}_{\sqrt{n}+1}$ at random.  In Lemma~\ref{l:est2} we generalize Lemma~\ref{l:est} to estimate $\hat{v}_i \approx \expect[S \sim \DC_{\sqrt{n} }|e_i \not \in S]{f_S(e_i)}$ with samples from $\mathcal{D}_{\sqrt{n} }$ and $\mathcal{D}_{\sqrt{n}+1}$.

\paragraph{The algorithm.}  We begin by computing the expected marginal contributions of  all elements.  We then place the elements in $3 \log n$ bins according to their estimated expected marginal contribution $\hat{v}_i $. The algorithm then simply returns either the best sample of size $k$ or a random subset of size $k$ of a random bin. Up to logarithmic factors, we can restrict our attention to just one bin.  We give a formal description below.

 \begin{algorithm}[H]
\caption{An $\tilde{\Omega}(n^{-1/4})$-optimization from samples algorithm over $\mathcal{D}^{sub}$ for monotone submodular functions.}
\label{alg:ubsub}
\begin{algorithmic}
    	\INPUT $\mathcal{S} = \{S_i \given (S_i, f(S_i)) \text{ is a sample} )\}$
    	\STATE With probability $\frac{1}{2}$:
    	\bindent 
    	\RETURN $\argmax_{S\in \mathcal{S}\given |S| =k} f(S)$ \COMMENT{best sample of size $k$}
    	\eindent
    	\STATE With probability $\frac{1}{2}$:
    	\bindent 
    	\STATE $(\hat{v}_1, \ldots, \hat{v}_n) \leftarrow \textsc{EEMC}(\mathcal{S})$
    	\STATE $\hat{v}_{max} \leftarrow \max_i \hat{v}_i$
    	\FOR {$j \in [3 \log n]$}
    	\STATE $B_j\leftarrow \left\{i \given \frac{\hat{v}_{max}}{2^{j_1-1}}\leq \hat{v}_i < \frac{\hat{v}_{max}}{2^{j_1}}\right\}$
    	\ENDFOR
    	\STATE Pick $j \in [3 \log n]$ u.a.r.
    	\RETURN $S$, a subset of $B_j$ of size $\min\{|B_j|, k\}$ u.a.r. \COMMENT{a random set from a random bin}
    	\eindent
  \end{algorithmic}
\end{algorithm}

 \paragraph{Analysis of the algorithm.}  The main crux of this result is in the analysis of the algorithm.

 \begin{restatable}{rThm}{tubsubmodular}
\label{t:ubsubmodular} Algorithm~\ref{alg:ubsub} is an $\tilde{\Omega}(n^{-1/4})$-\texttt{OPS} algorithm over $\mathcal{D}^{sub}$ for monotone submodular functions.
\end{restatable}

The analysis is divided in two cases, depending if a random set $S \sim \DC_{\sqrt{n} }$ of size $\sqrt{n}$ has low value or not. Let $S^{\star}$ be the optimal solution.
 
\begin{itemize}
\item   Assume that  $\expect[S \sim \mathcal{D}_{\sqrt{n}}]{f(S)} \leq f(S^{\star}) / 4$. Thus, optimal elements have large estimated expected marginal contribution $\hat{v}_i$  by submodularity.  Let $B^{\star}$ be the bin with the largest value among the bins with contributions $\hat{v} \geq f(S^{\star}) / (4k)$. We argue that a random subset of $B^{\star}$ of size $k$ performs well. Lemma~\ref{l:t/kn1/2}  shows that a random subset of $B^{\star}$ is a $|B^{\star}|  / (4k\sqrt{n})$-approximation. At a high level, a random subset $S$ of size $\sqrt{n}$ contains $|B^{\star}| /\sqrt{n}$ elements from bin $B^{\star}$ in expectation, and these $|B^{\star}|  /\sqrt{n}$ elements $S_{B^{\star}}$ have contributions at least $f(S^{\star}) / (4k)$ to $S_{B^{\star}}$. Lemma~\ref{l:k/t} shows that a random subset of $B^{\star}$ is an $\tilde{\Omega}(k/|B^{\star}| )$-approximation to $f(S^{\star})$. The proof first shows that $f( B^{\star})$ has high value by the assumption that a random set $S \sim \mathcal{D}_{\sqrt{n}}$ has low value, and then uses the fact that a subset of $B^{\star}$ of size $k$ is a $k/|B^{\star}| $ approximation to $B^{\star}$. Note that either  $|B^{\star}|   / (4k\sqrt{n})$ or $\tilde{\Omega}(k/|B^{\star}| )$ is at least  $\tilde{\Omega}(n^{-1/4})$.

\item Assume that $\expect[S \sim \mathcal{D}_{\sqrt{n}}]{f(S)} \geq f(S^{\star}) / 4$. We argue that the best sample of size $k$ performs well. Lemma~\ref{l:kn1/2} shows that, by submodularity, a random set of size $k$ is a $k / (4\sqrt{n})$ approximation since a random set of size $k$ is a fraction $k / (\sqrt{n})$ smaller than a random set from $\DC_{\sqrt{n}}$ in expectation. Lemma~\ref{l:1/k} shows that the best sample of size $k$ is a $1/k$-approximation since it contains the elements with the highest value with high probability. Note that either $k / (4\sqrt{n})$ or $1/k$ is at least $n^{-1/4}$.
\end{itemize}

\subsection{Bounded Curvature and Additive Functions}
\label{s:curv}
A simple $((1-c)^2 - o(1))$-\texttt{OPS} algorithms for monotone submodular functions with curvature $c$ over product distributions follows immediately from estimating expected marginal contributions. This result was recently improved to $(1-c)/(1+c-c^2)$, which was shown to be tight \cite{BS16}.  An immediate corollary is that additive (linear) functions, which have curvature $0$, are $(1 - o(1))$-\texttt{OPS} over product distributions.  The curvature $c$ of a function measures how far this function is to being additive. 

\begin{definition*} 
The \emph{curvature} $c$ of a submodular function $f$ is  $c := 1 - \min_{e \in N, S \subseteq N} f_{S \setminus e}(e) / f(e).$
\end{definition*}
This definition implies that $f_S(e) \geq (1 - c) f(e) \geq (1-c)f_T(e)$ for all $S,T$ and all $e \not \in S \cup T$ since  $f(e) \geq f(T \cup e) - f(T) = f_T(e)$ where the first inequality is by submodularity. The algorithm simply returns the $k$ elements with the highest expected marginal contributions.

 \begin{algorithm}[H]
\caption{\textsc{MaxMargCont}: A $((1-c)^2-o(1))$-optimization from samples algorithm for monotone submodular functions with curvature $c$.}
\label{a:curv}
\begin{algorithmic}
		\INPUT $\mathcal{S} = \{S_i \given (S_i, f(S_i)) \text{ is a sample} )\}$
    	\STATE $(\hat{v}_1, \ldots, \hat{v}_n) \leftarrow  \textsc{EEMC}(\mathcal{S})$ 
	\RETURN $S \leftarrow \argmax_{|T| = k} \sum_{i \in T} \hat{v}_i$ 
  \end{algorithmic}
\end{algorithm}

\begin{restatable}{rThm}{tcurv}
\label{t:curv} Let $f$ be a monotone submodular function with curvature $c$ and $\mathcal{D}$ be a product distribution with bounded marginals. Then \textsc{MaxMargCont} is a $((1-c)^2-o(1))$-\texttt{OPS} algorithm.
\end{restatable}
The proof follows almost immediately from the definition of curvature. Let $S$ be the set returned by the algorithm and $S^{\star}$ be the optimal solution, then $f(S)$ and $f(S^{\star})$ are sums of marginal contributions of elements in $S$ and $S^{\star}$ which are each at most a factor $1-c$ away from their estimated expected marginal contribution by curvature.  A $1 - o(1)$ approximation follows immediately for additive functions since they have curvature $c = 0$. A function $f$ is additive if  $f(S) = \sum_{e_i \in S} f(\{e_i\})$.
\begin{corollary}
\label{c:add}Let $f$ be an additive function and $\mathcal{D}$ be a product distribution with bounded marginals. Then \textsc{MaxMargCont} is a $(1-o(1))$-\texttt{OPS} algorithm.
\end{corollary}

\section{Recoverability}
\label{s:recoverability}

The largely negative results from the above sections lead to the question of how well must a function be learned for it to be optimizable from samples? One extreme is a notion we refer to as recoverability (\texttt{REC}).  A function is recoverable if it can be learned \emph{everywhere} within an approximation of $1 \pm 1/n^2$ from samples.  Does a function need to be learnable everywhere for it to be optimizable from samples?
\begin{definition*}
A function $f$ is recoverable for distribution $\mathcal{D}$ if there exists an algorithm which, given a polynomial number of samples drawn from $\mathcal{D}$, outputs a function $\tilde{f}$ such that for all sets $S$,
$$\left(1 - \frac{1}{n^2}\right) f(S) \leq \tilde{f}(S) \leq \left(1 +  \frac{1}{n^2} \right) f(S)$$
with  probability at least $1 - \delta$ over the samples and the randomness of the algorithm, where $\delta \in [0,1)$ is a constant.
\end{definition*}
This notion of recoverability is similar to the problem of approximating a function everywhere from \citet{GHIM09}. The differences are that recoverability is from samples whereas their setting allows value queries and that recoverability requires being within an approximation of $1 \pm 1/n^2$. It is important for us to be within such bounds and not  within some arbitrarily small constant because such perturbations can still lead to an $O(n^{-1/2 + \delta})$ impossibility result for optimization \cite{HS15}. We show that if a monotone submodular function $f$ is recoverable then it is optimizable from samples by using the greedy algorithm on the recovered function $\tilde{f}$. The proof is similar to the classical analysis of the greedy algorithm.

\begin{restatable}{rThm}{greedyrecThm}
\label{t:greedyrec}  If a monotone submodular function $f$  is recoverable over $\mathcal{D}$, then it is $1-1/e -o(1)$-optimizable from samples over $\mathcal{D}$. For additive functions, it is $1 - o(1)$-optimizable from samples.
\end{restatable}

We show that additive functions are in \texttt{REC} under some mild condition. Combined with the previous result, we get an alternate proof from the previous section for additive functions being $1 - o(1)$-optimizable from samples over product distributions. 

\begin{restatable}{rLem}{laddrec}
\label{l:addrec}
Let $f$ be an additive function with $v_{max} = \max_i f(\{e_i\})$, $v_{min} = \min_i f(\{e_i\})$  and let $\mathcal{D}$ be a product distribution with bounded marginals. If $v_{min} \geq  v_{max} / poly(n)$, then $f$ is recoverable for $\mathcal{D}$.
\end{restatable}

We also note that submodular functions that are a $c$-\emph{junta} for some constant $c$ are recoverable. A function $f$  is a $c$-junta \cite{mossel2003learning,FV13-submodular_juntas,valiant2012finding} if it depends only on a set of elements $T$ of size $c$. If $c$ is constant, then with enough samples, $T$ can be learned since each element not in $T$ is in at least one sample which does not contain any element in $T$. For each subset of $T$, there is also at least one sample which intersects with $T$ in exactly that subset, so $f$ is exactly recoverable.

The previous results lead us to the following question: Does a function need to be recoverable to be optimizable from samples? We show that it is not the case since unit demand functions are optimizable from samples and not recoverable. A function $f$ is a unit demand function if $f(S) = \max_{e_i \in S} f(\{e_i\})$.

\begin{restatable}{rLem}{lud}
\label{l:ud}
Unit demand functions are not recoverable for $k \geq n^{\epsilon}$ but are $1$-\texttt{OPS}.
\end{restatable}

 We conclude that functions do not need to be learnable everywhere from samples to be optimizable from samples.

\section{Additional related work}
\label{s:addrelatedwork}
\paragraph{Revenue maximization from samples.} The discrepancy between the model on which algorithms optimize and the true state of nature has recently been studied in algorithmic mechanism design.
Most closely related to our work are several recent papers (e.g. \cite{CR14, DHN14, CHN15, HMR15, MR15-auction-pseudo, CCDEHT15-quasipoly_signaling}) that also consider models that bypass the learning algorithm and let the mechanism designer access samples from a distribution rather than an explicit Bayesian prior.
In contrast to our negative conclusions, these papers achieve mostly positive results.
In particular, Huang et al. \cite{HMR15} show that the obtainable revenue is much closer to the optimum than the information-theoretic bound on learning the valuation distribution.

\paragraph{Comparison to online learning and reinforcement learning.} Another line of work which combines decision-making and learning is online learning  (see survey \cite{onlinelearning}). In online learning, a player iteratively makes decisions. For each decision, the player incurs a cost and the cost function for the current iteration is immediately revealed. The objective is to minimize regret, which is the difference between the sum of the costs of the decisions of the player and the sum of the costs of the best fixed decision. The fundamental differences with our framework are that decisions are made online after each observation, instead of offline given a collection of observations. The benchmarks, regret in one case and the optimal solution in the other, are not comparable.

A similar comparison can be made with the problem of reinforcement learning, where at each iteration the player typically interacts with a Markov decision process (MDP) \cite{papadimitriou1987complexity}. At each iteration, an action is chosen in an online manner and the player receives a reward based on the action and the state in the MDP she is in. Again, this differs from our setting where there is one offline decision to be made given a collection observations.

\paragraph{Additional learning results for submodular functions.} In addition to the \texttt{PMAC} learning results  mentioned in the introduction for coverage functions, there are multiple learning results for submodular functions.  Monotone submodular functions are $\alpha$-\texttt{PMAC} learnable over product distributions for some constant $\alpha$ under some assumptions \cite{BH11-PMAC}.  Impossibility results arise for general distributions, in which case submodular functions are not  $\tilde{\Omega}(n^{-1/3})$-\texttt{PMAC} learnable \cite{BH11-PMAC}. Finally, submodular functions can be $(1-\epsilon)$-\texttt{PMAC} learned for the uniform distribution over all sets with a running time and sample complexity exponential in $\epsilon$ and polynomial in $n$ \cite{FV13-submodular_juntas}. This exponential dependency is necessary since $2^{\Omega(\epsilon^{-2/3})}$ samples are needed to learn submodular functions with $\ell_1$-error of $\epsilon$ over this distribution \cite{FKV13-submodular_trees}.

\newpage


\bibliographystyle{plainnat}
\bibliography{sampleOpt}

\begin{thebibliography}{60}
\providecommand{\natexlab}[1]{#1}
\providecommand{\url}[1]{\texttt{#1}}
\expandafter\ifx\csname urlstyle\endcsname\relax
  \providecommand{\doi}[1]{doi: #1}\else
  \providecommand{\doi}{doi: \begingroup \urlstyle{rm}\Url}\fi

\bibitem[Antonellis et~al.(2012)Antonellis, Das~Sarma, and
  Dughmi]{antonellis2012dynamic}
Ioannis Antonellis, Anish Das~Sarma, and Shaddin Dughmi.
\newblock Dynamic covering for recommendation systems.
\newblock In \emph{Proceedings of the 21st ACM international conference on
  Information and knowledge management}, pages 26--34. ACM, 2012.

\bibitem[Badanidiyuru et~al.()Badanidiyuru, Dobzinski, Fu, Kleinberg, Nisan,
  and Roughgarden]{BDFKNR12-sketches}
Ashwinkumar Badanidiyuru, Shahar Dobzinski, Hu~Fu, Robert Kleinberg, Noam
  Nisan, and Tim Roughgarden.
\newblock Sketching valuation functions.
\newblock In \emph{Proceedings of the Twenty-Third Annual {ACM-SIAM} Symposium
  on Discrete Algorithms, {SODA} 2012, Kyoto, Japan, January 17-19, 2012}.

\bibitem[Balcan(2015)]{Balcan15-AAMAS}
Maria{-}Florina Balcan.
\newblock Learning submodular functions with applications to multi-agent
  systems.
\newblock In \emph{Proceedings of the 2015 International Conference on
  Autonomous Agents and Multiagent Systems, {AAMAS} 2015, Istanbul, Turkey, May
  4-8, 2015}, 2015.

\bibitem[Balcan and Harvey(2011)]{BH11-PMAC}
Maria{-}Florina Balcan and Nicholas J.~A. Harvey.
\newblock Learning submodular functions.
\newblock In \emph{Proceedings of the 43rd {ACM} Symposium on Theory of
  Computing, {STOC} 2011, San Jose, CA, USA, 6-8 June 2011}, 2011.

\bibitem[Balcan et~al.(2012)Balcan, Constantin, Iwata, and
  Wang]{BCIW12-valuations}
Maria{-}Florina Balcan, Florin Constantin, Satoru Iwata, and Lei Wang.
\newblock Learning valuation functions.
\newblock In \emph{{COLT} 2012 - The 25th Annual Conference on Learning Theory,
  June 25-27, 2012, Edinburgh, Scotland}, 2012.

\bibitem[Balkanski et~al.(2016)Balkanski, Rubinstein, and Singer]{BS16}
Eric Balkanski, Aviad Rubinstein, and Yaron Singer.
\newblock The power of optimization from samples.
\newblock In \emph{NIPS}, 2016.

\bibitem[Blake and Studholme(2006)]{BS06}
Ian~F Blake and Chris Studholme.
\newblock Properties of random matrices and applications.
\newblock \emph{Unpublished report available at http://www. cs. toronto.
  edu/\~{} cvs/coding}, 2006.

\bibitem[Borgs et~al.(2014)Borgs, Brautbar, Chayes, and
  Lucier]{borgs2014maximizing}
Christian Borgs, Michael Brautbar, Jennifer Chayes, and Brendan Lucier.
\newblock Maximizing social influence in nearly optimal time.
\newblock In \emph{Proceedings of the Twenty-Fifth Annual ACM-SIAM Symposium on
  Discrete Algorithms}, pages 946--957. SIAM, 2014.

\bibitem[Buchfuhrer et~al.(2010)Buchfuhrer, Schapira, and
  Singer]{buchfuhrer2010computation}
Dave Buchfuhrer, Michael Schapira, and Yaron Singer.
\newblock Computation and incentives in combinatorial public projects.
\newblock In \emph{Proceedings of the 11th ACM conference on Electronic
  commerce}, pages 33--42. ACM, 2010.

\bibitem[Chakrabarty and Huang(2012)]{chakrabarty2012testing}
Deeparnab Chakrabarty and Zhiyi Huang.
\newblock Testing coverage functions.
\newblock In \emph{Automata, Languages, and Programming}, pages 170--181. 2012.

\bibitem[Chawla et~al.(2014)Chawla, Hartline, and Nekipelov]{CHN15}
Shuchi Chawla, Jason~D. Hartline, and Denis Nekipelov.
\newblock Mechanism design for data science.
\newblock In \emph{{ACM} Conference on Economics and Computation, {EC} '14,
  Stanford , CA, USA, June 8-12, 2014}, pages 711--712, 2014.

\bibitem[Cheng et~al.(2015)Cheng, Cheung, Dughmi, Emamjomeh-Zadeh, Han, and
  Teng]{CCDEHT15-quasipoly_signaling}
Yu~Cheng, Ho~Yee Cheung, Shaddin Dughmi, Ehsan Emamjomeh-Zadeh, Li~Han, and
  Shang-Hua Teng.
\newblock Mixture selection, mechanism design, and signaling.
\newblock In \emph{FOCS}, 2015.
\newblock To appear.

\bibitem[Chierichetti et~al.(2010)Chierichetti, Kumar, and
  Tomkins]{chierichetti2010max}
Flavio Chierichetti, Ravi Kumar, and Andrew Tomkins.
\newblock Max-cover in map-reduce.
\newblock In \emph{Proceedings of the 19th international conference on World
  wide web}, pages 231--240. ACM, 2010.

\bibitem[Cole and Roughgarden(2014)]{CR14}
Richard Cole and Tim Roughgarden.
\newblock The sample complexity of revenue maximization.
\newblock In \emph{Proceedings of the 46th Annual ACM Symposium on Theory of
  Computing}, pages 243--252, 2014.

\bibitem[Daneshmand et~al.(2014)Daneshmand, Gomez{-}Rodriguez, Song, and
  Sch{\"{o}}lkopf]{DGSS14}
Hadi Daneshmand, Manuel Gomez{-}Rodriguez, Le~Song, and Bernhard
  Sch{\"{o}}lkopf.
\newblock Estimating diffusion network structures: Recovery conditions, sample
  complexity {\&} soft-thresholding algorithm.
\newblock In \emph{Proceedings of the 31th International Conference on Machine
  Learning, {ICML} 2014, Beijing, China, 21-26 June 2014}, pages 793--801,
  2014.

\bibitem[Dasgupta et~al.(2007)Dasgupta, Ghosh, Kumar, Olston, Pandey, and
  Tomkins]{dasgupta2007discoverability}
Anirban Dasgupta, Arpita Ghosh, Ravi Kumar, Christopher Olston, Sandeep Pandey,
  and Andrew Tomkins.
\newblock The discoverability of the web.
\newblock In \emph{Proceedings of the 16th international conference on World
  Wide Web}, pages 421--430. ACM, 2007.

\bibitem[Dobzinski and Schapira(2006)]{dobzinski2006improved}
Shahar Dobzinski and Michael Schapira.
\newblock An improved approximation algorithm for combinatorial auctions with
  submodular bidders.
\newblock In \emph{Proceedings of the seventeenth annual ACM-SIAM symposium on
  Discrete algorithm}, pages 1064--1073. Society for Industrial and Applied
  Mathematics, 2006.

\bibitem[Du et~al.(2013)Du, Song, Gomez{-}Rodriguez, and Zha]{DSGZ13}
Nan Du, Le~Song, Manuel Gomez{-}Rodriguez, and Hongyuan Zha.
\newblock Scalable influence estimation in continuous-time diffusion networks.
\newblock In \emph{Advances in Neural Information Processing Systems 26: 27th
  Annual Conference on Neural Information Processing Systems 2013. Proceedings
  of a meeting held December 5-8, 2013, Lake Tahoe, Nevada, United States.},
  pages 3147--3155, 2013.

\bibitem[Du et~al.(2014{\natexlab{a}})Du, Liang, Balcan, and Song]{DBS14b}
Nan Du, Yingyu Liang, Maria{-}Florina Balcan, and Le~Song.
\newblock Learning time-varying coverage functions.
\newblock In \emph{Advances in Neural Information Processing Systems 27: Annual
  Conference on Neural Information Processing Systems 2014, December 8-13 2014,
  Montreal, Quebec, Canada}, pages 3374--3382, 2014{\natexlab{a}}.

\bibitem[Du et~al.(2014{\natexlab{b}})Du, Liang, Balcan, and Song]{DLBS14a}
Nan Du, Yingyu Liang, Maria{-}Florina Balcan, and Le~Song.
\newblock Influence function learning in information diffusion networks.
\newblock In \emph{Proceedings of the 31th International Conference on Machine
  Learning, {ICML} 2014, Beijing, China, 21-26 June 2014}, pages 2016--2024,
  2014{\natexlab{b}}.

\bibitem[Du et~al.(2014{\natexlab{c}})Du, Liang, Balcan, and
  Song]{du2014learning}
Nan Du, Yingyu Liang, Maria-Florina~F Balcan, and Le~Song.
\newblock Learning time-varying coverage functions.
\newblock In \emph{Advances in neural information processing systems}, pages
  3374--3382, 2014{\natexlab{c}}.

\bibitem[Dughmi(2011)]{dughmi2011truthful}
Shaddin Dughmi.
\newblock A truthful randomized mechanism for combinatorial public projects via
  convex optimization.
\newblock In \emph{Proceedings of the 12th ACM conference on Electronic
  commerce}, pages 263--272. ACM, 2011.

\bibitem[Dughmi and Vondr{\'a}k(2015)]{dughmi2015limitations}
Shaddin Dughmi and Jan Vondr{\'a}k.
\newblock Limitations of randomized mechanisms for combinatorial auctions.
\newblock \emph{Games and Economic Behavior}, 92:\penalty0 370--400, 2015.

\bibitem[Dughmi et~al.(2011)Dughmi, Roughgarden, and Yan]{dughmi2011convex}
Shaddin Dughmi, Tim Roughgarden, and Qiqi Yan.
\newblock From convex optimization to randomized mechanisms: toward optimal
  combinatorial auctions.
\newblock In \emph{Proceedings of the forty-third annual ACM symposium on
  Theory of computing}, pages 149--158. ACM, 2011.

\bibitem[Dughmi et~al.(2014)Dughmi, Han, and Nisan]{DHN14}
Shaddin Dughmi, Li~Han, and Noam Nisan.
\newblock Sampling and representation complexity of revenue maximization.
\newblock In \emph{Web and Internet Economics - 10th International Conference,
  {WINE} 2014, Beijing, China, December 14-17, 2014. Proceedings}, pages
  277--291, 2014.

\bibitem[Ene et~al.(2013)Ene, Vondr{\'{a}}k, and Wu]{EneVW13}
Alina Ene, Jan Vondr{\'{a}}k, and Yi~Wu.
\newblock Local distribution and the symmetry gap: Approximability of multiway
  partitioning problems.
\newblock In \emph{Proceedings of the Twenty-Fourth Annual {ACM-SIAM} Symposium
  on Discrete Algorithms, {SODA} 2013, New Orleans, Louisiana, USA, January
  6-8, 2013}, pages 306--325, 2013.

\bibitem[Feige(1998)]{feige1998threshold}
Uriel Feige.
\newblock A threshold of ln n for approximating set cover.
\newblock \emph{Journal of the ACM (JACM)}, 45\penalty0 (4):\penalty0 634--652,
  1998.

\bibitem[Feldman and Kothari(2014)]{FK14-coverage}
Vitaly Feldman and Pravesh Kothari.
\newblock Learning coverage functions and private release of marginals.
\newblock In \emph{Proceedings of The 27th Conference on Learning Theory,
  {COLT} 2014, Barcelona, Spain, June 13-15, 2014}, 2014.

\bibitem[Feldman and Vondr{\'{a}}k(2013)]{FV13-submodular_juntas}
Vitaly Feldman and Jan Vondr{\'{a}}k.
\newblock Optimal bounds on approximation of submodular and {XOS} functions by
  juntas.
\newblock In \emph{54th Annual {IEEE} Symposium on Foundations of Computer
  Science, {FOCS} 2013, 26-29 October, 2013, Berkeley, CA, {USA}}, 2013.

\bibitem[Feldman and Vondr{\'a}k(2015)]{feldman2015tight}
Vitaly Feldman and Jan Vondr{\'a}k.
\newblock Tight bounds on low-degree spectral concentration of submodular and
  xos functions.
\newblock In \emph{Foundations of Computer Science (FOCS), 2015 IEEE 56th
  Annual Symposium on}, pages 923--942. IEEE, 2015.

\bibitem[Feldman et~al.(2013)Feldman, Kothari, and
  Vondr{\'{a}}k]{FKV13-submodular_trees}
Vitaly Feldman, Pravesh Kothari, and Jan Vondr{\'{a}}k.
\newblock Representation, approximation and learning of submodular functions
  using low-rank decision trees.
\newblock In \emph{{COLT} 2013 - The 26th Annual Conference on Learning Theory,
  June 12-14, 2013, Princeton University, NJ, {USA}}, 2013.

\bibitem[Goemans et~al.(2009)Goemans, Harvey, Iwata, and Mirrokni]{GHIM09}
Michel~X Goemans, Nicholas~JA Harvey, Satoru Iwata, and Vahab Mirrokni.
\newblock Approximating submodular functions everywhere.
\newblock In \emph{Proceedings of the twentieth Annual ACM-SIAM Symposium on
  Discrete Algorithms}, 2009.

\bibitem[Gomez{-}Rodriguez et~al.(2010)Gomez{-}Rodriguez, Leskovec, and
  Krause]{GLK10}
Manuel Gomez{-}Rodriguez, Jure Leskovec, and Andreas Krause.
\newblock Inferring networks of diffusion and influence.
\newblock In \emph{Proceedings of the 16th {ACM} {SIGKDD} International
  Conference on Knowledge Discovery and Data Mining, Washington, DC, USA, July
  25-28, 2010}, pages 1019--1028, 2010.

\bibitem[Gomez~Rodriguez et~al.(2010)Gomez~Rodriguez, Leskovec, and
  Krause]{gomez2010inferring}
Manuel Gomez~Rodriguez, Jure Leskovec, and Andreas Krause.
\newblock Inferring networks of diffusion and influence.
\newblock In \emph{Proceedings of the 16th ACM SIGKDD international conference
  on Knowledge discovery and data mining}, pages 1019--1028. ACM, 2010.

\bibitem[Guestrin et~al.(2005)Guestrin, Krause, and Singh]{guestrin2005near}
Carlos Guestrin, Andreas Krause, and Ajit~Paul Singh.
\newblock Near-optimal sensor placements in gaussian processes.
\newblock In \emph{Proceedings of the 22nd international conference on Machine
  learning}, pages 265--272, 2005.

\bibitem[Gupta et~al.(2013)Gupta, Hardt, Roth, and Ullman]{gupta2013privately}
Anupam Gupta, Moritz Hardt, Aaron Roth, and Jonathan Ullman.
\newblock Privately releasing conjunctions and the statistical query barrier.
\newblock \emph{SIAM Journal on Computing}, 42\penalty0 (4):\penalty0
  1494--1520, 2013.

\bibitem[Hassidim and Singer(2015)]{HS15}
Avinatan Hassidim and Yaron Singer.
\newblock Submodular optimization under noise.
\newblock 2015.
\newblock Working paper.

\bibitem[Hazan(2016)]{onlinelearning}
Elad Hazan.
\newblock Draft: Introduction to online convex optimization.
\newblock In \emph{Foundations and Trends in Optimization, vol. XX, no. XX},
  pages 1–--172. 2016.

\bibitem[He and Kempe(2016)]{HK16}
Xinran He and David Kempe.
\newblock Robust influence maximization.
\newblock In \emph{Proceedings of the 22nd {ACM} {SIGKDD} International
  Conference on Knowledge Discovery and Data Mining, San Francisco, CA, USA,
  August 13-17, 2016}, pages 885--894, 2016.

\bibitem[Huang et~al.(2015)Huang, Mansour, and Roughgarden]{HMR15}
Zhiyi Huang, Yishay Mansour, and Tim Roughgarden.
\newblock Making the most of your samples.
\newblock In \emph{Proceedings of the Sixteenth {ACM} Conference on Economics
  and Computation, {EC} '15, Portland, OR, USA, June 15-19, 2015}, pages
  45--60, 2015.

\bibitem[Kempe et~al.(2003)Kempe, Kleinberg, and Tardos]{kempe2003maximizing}
David Kempe, Jon Kleinberg, and {\'E}va Tardos.
\newblock Maximizing the spread of influence through a social network.
\newblock In \emph{Proceedings of the ninth ACM SIGKDD international conference
  on Knowledge discovery and data mining}, pages 137--146. ACM, 2003.

\bibitem[Korula et~al.(2015)Korula, Mirrokni, and
  Zadimoghaddam]{KMZ15-online_SWM}
Nitish Korula, Vahab~S. Mirrokni, and Morteza Zadimoghaddam.
\newblock Online submodular welfare maximization: Greedy beats 1/2 in random
  order.
\newblock In \emph{Proceedings of the Forty-Seventh Annual {ACM} on Symposium
  on Theory of Computing, {STOC} 2015, Portland, OR, USA, June 14-17, 2015},
  2015.

\bibitem[Krause and Guestrin(2007)]{krause2007near}
Andreas Krause and Carlos Guestrin.
\newblock Near-optimal observation selection using submodular functions.
\newblock In \emph{AAAI}, volume~7, pages 1650--1654, 2007.

\bibitem[Lehmann et~al.(2001)Lehmann, Lehmann, and
  Nisan]{lehmann2001combinatorial}
Benny Lehmann, Daniel Lehmann, and Noam Nisan.
\newblock Combinatorial auctions with decreasing marginal utilities.
\newblock In \emph{Proceedings of the 3rd ACM conference on Electronic
  Commerce}, pages 18--28. ACM, 2001.

\bibitem[Lin and Bilmes(2011)]{lin2011class}
Hui Lin and Jeff Bilmes.
\newblock A class of submodular functions for document summarization.
\newblock In \emph{Proceedings of the 49th Annual Meeting of the Association
  for Computational Linguistics: Human Language Technologies-Volume 1}, pages
  510--520. Association for Computational Linguistics, 2011.

\bibitem[Mirrokni et~al.(2008)Mirrokni, Schapira, and Vondr{\'a}k]{MSV08}
Vahab Mirrokni, Michael Schapira, and Jan Vondr{\'a}k.
\newblock Tight information-theoretic lower bounds for welfare maximization in
  combinatorial auctions.
\newblock In \emph{Proceedings of the 9th ACM conference on Electronic
  commerce}, 2008.

\bibitem[Morgenstern and Roughgarden(2015)]{MR15-auction-pseudo}
Jamie Morgenstern and Tim Roughgarden.
\newblock The pseudo-dimension of nearly-optimal auctions.
\newblock In \emph{NIPS}, page Forthcoming, 12 2015.
\newblock URL \url{papers/auction-pseudo.pdf}.

\bibitem[Mossel et~al.(2003)Mossel, O'Donnell, and
  Servedio]{mossel2003learning}
Elchanan Mossel, Ryan O'Donnell, and Rocco~P Servedio.
\newblock Learning juntas.
\newblock In \emph{Proceedings of the thirty-fifth annual ACM symposium on
  Theory of computing}, pages 206--212. ACM, 2003.

\bibitem[Narasimhan et~al.(2015)Narasimhan, Parkes, and Singer]{NPS15}
Harikrishna Narasimhan, David~C. Parkes, and Yaron Singer.
\newblock Learnability of influence in networks.
\newblock In \emph{Advances in Neural Information Processing Systems 28: Annual
  Conference on Neural Information Processing Systems 2015, December 7-12,
  2015, Montreal, Quebec, Canada}, pages 3186--3194, 2015.

\bibitem[Nemhauser et~al.(1978)Nemhauser, Wolsey, and Fisher]{NWF78}
G.~L. Nemhauser, L.~A. Wolsey, and M.~L. Fisher.
\newblock An analysis of approximations for maximizing submodular set functions
  ii.
\newblock \emph{Math. Programming Study 8}, 1978.

\bibitem[Papadimitriou and Tsitsiklis(1987)]{papadimitriou1987complexity}
Christos~H Papadimitriou and John~N Tsitsiklis.
\newblock The complexity of markov decision processes.
\newblock \emph{Mathematics of operations research}, 12\penalty0 (3):\penalty0
  441--450, 1987.

\bibitem[Saha and Getoor(2009)]{saha2009maximum}
Barna Saha and Lise Getoor.
\newblock On maximum coverage in the streaming model \& application to
  multi-topic blog-watch.
\newblock In \emph{SDM}, volume~9, pages 697--708. SIAM, 2009.

\bibitem[Seeman and Singer(2013)]{seeman2013adaptive}
Lior Seeman and Yaron Singer.
\newblock Adaptive seeding in social networks.
\newblock In \emph{Foundations of Computer Science (FOCS), 2013 IEEE 54th
  Annual Symposium on}, pages 459--468. IEEE, 2013.

\bibitem[Singer(2012)]{singer2012win}
Yaron Singer.
\newblock How to win friends and influence people, truthfully: influence
  maximization mechanisms for social networks.
\newblock In \emph{Proceedings of the fifth ACM international conference on Web
  search and data mining}, pages 733--742. ACM, 2012.

\bibitem[Swaminathan et~al.(2009)Swaminathan, Mathew, and
  Kirovski]{swaminathan2009essential}
Ashwin Swaminathan, Cherian~V Mathew, and Darko Kirovski.
\newblock Essential pages.
\newblock In \emph{Proceedings of the 2009 IEEE/WIC/ACM International Joint
  Conference on Web Intelligence and Intelligent Agent Technology-Volume 01},
  pages 173--182, 2009.

\bibitem[Takamura and Okumura(2009)]{takamura2009text}
Hiroya Takamura and Manabu Okumura.
\newblock Text summarization model based on maximum coverage problem and its
  variant.
\newblock In \emph{Proceedings of the 12th Conference of the European Chapter
  of the Association for Computational Linguistics}, pages 781--789.
  Association for Computational Linguistics, 2009.

\bibitem[Valiant(2012)]{valiant2012finding}
Gregory Valiant.
\newblock Finding correlations in subquadratic time, with applications to
  learning parities and juntas.
\newblock In \emph{Foundations of Computer Science (FOCS), 2012 IEEE 53rd
  Annual Symposium on}, pages 11--20. IEEE, 2012.

\bibitem[Valiant(1984)]{Valiant84-PAC}
Leslie~G. Valiant.
\newblock {A Theory of the Learnable}.
\newblock \emph{Commun. {ACM}}, 1984.

\bibitem[Vondr{\'{a}}k(2013)]{Vondrak13}
Jan Vondr{\'{a}}k.
\newblock Symmetry and approximability of submodular maximization problems.
\newblock \emph{{SIAM} J. Comput.}, 42\penalty0 (1):\penalty0 265--304, 2013.

\bibitem[Yue and Joachims(2008)]{yue2008predicting}
Yisong Yue and Thorsten Joachims.
\newblock Predicting diverse subsets using structural svms.
\newblock In \emph{Proceedings of the 25th international conference on Machine
  learning}, pages 1224--1231, 2008.

\end{thebibliography}

\newpage
\section*{Appendix}
\setcounter{section}{0}
\renewcommand{\thesection}{\Alph{section}}
\newpage
\section{Impossibility of OPS}
\label{s:mpimpossibility}
\subsection*{A Framework for OPS Hardness}

We reduce the problem of showing hardness results to the problem of constructing $g,b,m,m^+$ with an $(\alpha, \beta)$-gap. Recall that a partition $P$ has $r$ parts $T_1, \ldots, T_r$ of $k$  elements and a part $M$ of remaining $n - rk$ elements. The functions $f^{P,i}(S) \in \mathcal{F}(g,b,m,m^+) $ are defined as $f^{P,i}(S) := (1 - m(S \cap M)) \left(g(S \cap T_i) + b(S \cap T_{-i})\right) + m^+(S \cap  M)$ with $i \in [r]$.  
	
\tgap*
\begin{proof} Fix any distribution $\mathcal{D}$.  We first claim that for a fixed set $S$, $f^{P,i}(S) $ is independent of $i$  with probability $1 - n^{-\omega(1)}$ over a uniformly random partition $P \sim \mathcal{U}(\PC)$.  If $|S| \leq t$, then the claim holds immediately by the identical on small samples property. If $|S| \geq t$, then $m(S \cap M) = 1$ with probability $1 - n^{-\omega(1)}$ over $P$ by the masking on large samples property and $f^{P,i}(S) = m^+(S \cap M)$. 

Next, we claim that there exists a partition $P \in \mathcal{P}$ such that $f^{P,i}(S)$ is independent of $i$  with probability $1 - n^{-\omega(1)}$ over $S \sim \mathcal{D}$. Denote the event that $f^{P,i}(S)$ is independent of $i$ by $I(S,P)$.  By switching sums,
\begin{align*}
& \sum_{P \in \PC} \prob{P \sim \mathcal{U}(\PC)} \sum_{S \in 2^N}\prob{S \sim \DC} \mathds{1}_{I(S,P)} \\
 = & \sum_{S \in 2^N}\prob{S \sim \DC} \sum_{P \in \PC} \prob{P \sim \mathcal{U}(\PC)} \mathds{1}_{I(S,P)} \\
  \geq & \sum_{S \in 2^N}\prob{S \sim \DC} \left(1 - n^{-\omega(1)}\right) \\
    = & 1 - n^{-\omega(1)}
\end{align*}
where the inequality is by the first claim. Thus there exists some $P$ such that $$ \sum_{S \in 2^N}\prob{S \sim \DC}\mathds{1}_{I(S,P)} \geq 1- n^{-\omega(1)},$$ which proves the desired claim.

Fix a partition $P$ such that the previous claim holds, i.e., $f^{P,i}(S)$ is independent of $i$ with probability $1 - n^{-\omega(1)}$ over a sample $S \sim \DC$ . Then, by a union bound over the polynomially many samples,  $f^{P,i}(S)$ is independent of $i$ for all samples $S$ with probability $1 - n^{-\omega(1)}$, and we assume this is the case for the remaining of the proof. It follows  that the choices of the algorithm given samples from $f \in \{f^{P,i}\}_{i=1}^{r}$  are independent of $i$. Pick $i \in [r]$ uniformly at random and consider the (possibly randomized) set $S$ returned by the algorithm. Since $S$ is independent of $i$, we get  $\expect[i,S]{|S \cap T_i|} \leq k/r$. Let $S_{k/r} = \argmax_{S:|S|=k/r}(g(S))$,  we obtain
 \begin{align*}
 \expect[i, S]{f^{P,i}(S)} &  \leq  \expect[i, S]{ g(S \cap T_i) + b(S \cap T_{- i})+ m^+(S \cap M)} &&(m(S \cap M) \leq 1) \\
& \leq  g(S_{k/r}) + b(S)+ m^+(S) &&\text{(monotone and submodular)} \\
& \leq  \frac{  1 }{r ( 1- \beta)}  g(T_i) + \frac{2}{\alpha} g(T_i) &&\text{(curvature and gap)} \\
& \leq  2\max\left(  \frac{1}{r ( 1- \beta)}, \frac{2}{\alpha}\right) f^{P,i}(T_i)
\end{align*}
Thus, there exists at least one $i$ such that the algorithm does not obtain a $2\max(  1 /(r ( 1- \beta)), 2/\alpha)$-approximation to $f^{\mathcal{P},i}(T_i)$, and $T_i$ is the optimal solution.
\end{proof}

\subsection*{OPS Hardness of Coverage Functions}
We consider the cardinality constraint $k = n^{2/5 - \epsilon}$ and the number of parts $r = n^{1/5 - \epsilon}$.
\subsubsection*{Construction the Good and the Bad Coverage Functions}

For symmetric functions $h$ (i.e. whose value only depends on the size of the set), we abuse notation and simply write $h(y)$ instead of $h(S)$ for a set $S$ of size $y$. We begin by showing that  the primitives $C_{p}(y) = p \cdot \left( 1 - (1 - 1/p)^y \right)$ (illustrated in Figure~\ref{fig:graph}) are coverage functions. It then follows that the functions $g$ and $b$ are coverage.
\cct*
\begin{proof}
Note that 
\begin{align*}
C_{p}(S) & =  \sum_{T \given |T \cap S| \geq 1} w(a_T) \\ 
& = t \cdot \sum_{T \given |T \cap S| \geq 1} \prob{T \sim \B(N, 1/p)}\\ 
& = t \cdot \left(1 - \sum_{T \given |T \cap S| = 0} \prob{T \sim \B(N, 1/p)} \right)\\ 
& = t \cdot \left(1 - \Pr_{T \sim \B(N, 1/p)}[|T \cap S| = 0]\right)\\ 
& = t \cdot \left(1 - \left(1 - \frac{1}{p}\right)^{|S|}\right).
\end{align*}
\end{proof}

\begin{figure}
\centering
\includegraphics[width=7cm]{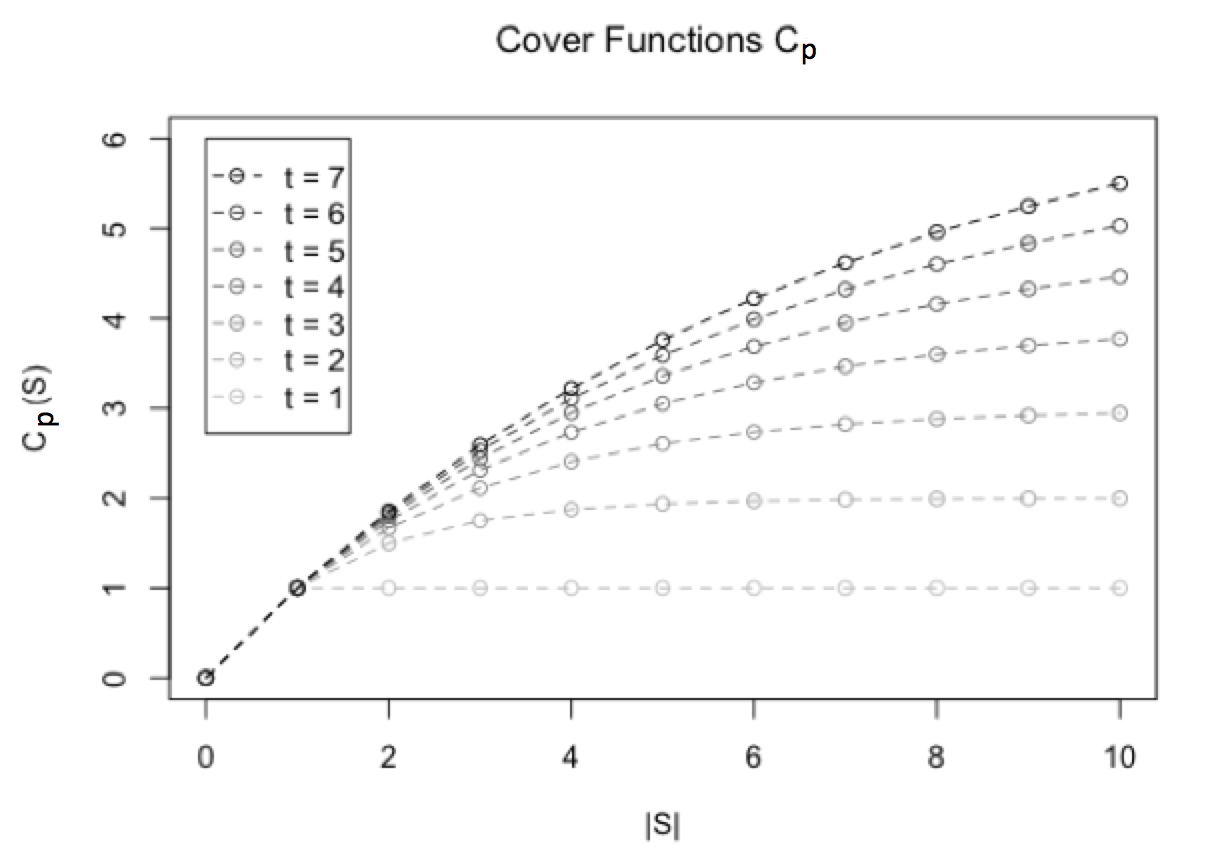}
\caption{The value of coverage functions $C_{p}(y)$ for $1 \leq p \leq 7$ and sets $y \in [10]$.}
\label{fig:graph}
\end{figure}

In the remaining of this section, we prove Lemma~\ref{l:abconstruction}.

\labconstruction*

The good and bad functions are defined as $g(y) = y + \sum_{j \given x_j < 0} (-x_j) C_{p_j}(y)$ and $b(S) = \sum_{j=1, j \neq i}^r b'(S \cap T_j)$ with
$b'(y) := \sum_{j \given x_j > 0} x_j C_{p_j}(y)$. We obtain the coefficients $\bx$ by solving the system of linear equations $A \bx = \by$ where $A_{ij} := C_{p_j}(i)$ and $y_j := j$ as illustrated in Figure~\ref{f:matrix}, with $i,j \in [\ell]$.

\begin{figure}[t]
\centering
  \includegraphics[width=.5\linewidth]{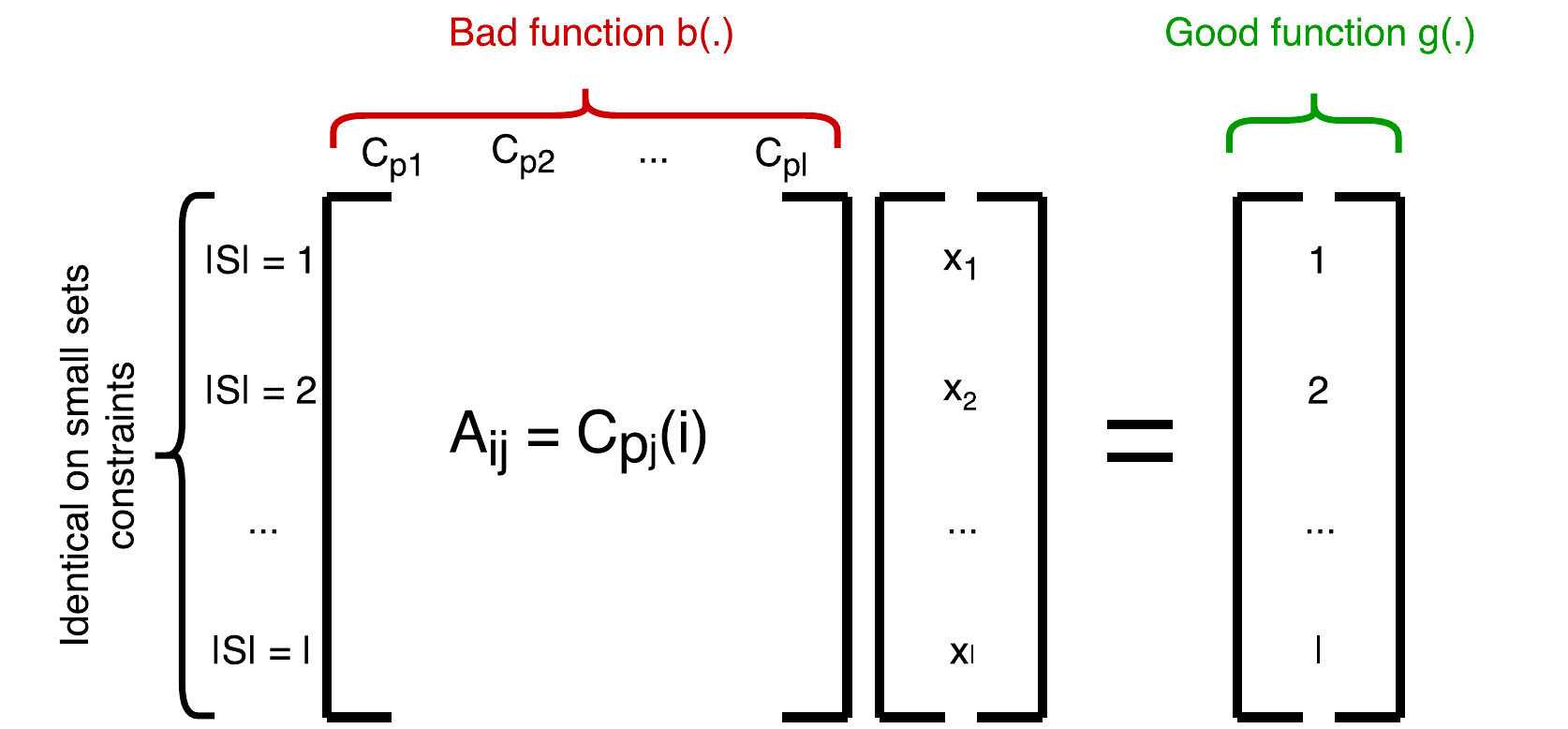}
\caption{The matrix $A$.}
\label{f:matrix}
\end{figure}

To prove Lemma~\ref{l:abconstruction}, we begin by showing that $A$ is invertible in Lemma~\ref{lem:invertible}, so that the coefficients $\bx$ satisfying the system of linear equations exist. We then show the three desired properties. Lemma~\ref{lem:skewed-concentration} shows that a set $S$ of size at most $n^{3/5 + \epsilon}$ contains at most $\ell$ elements from any part $T_j$ w.p. $1 - n^{-\omega(1)}$, thus the identical on small samples property holds by the system of linear equations (Lemma~\ref{l:issp}). Lemma~\ref{l:boundCoeff} bounds the coefficients $\bx $, thus the $y$ term in the good function dominates and we obtain the gap (Lemma~\ref{l:gap}) and curvature (Lemma~\ref{l:curv}) properties.

\begin{lemma}
\label{lem:invertible}
Matrix $A(\{p_j\}_{j=1}^\ell)$ is invertible for some set of integers  $\{p_j\}_{j=1}^\ell$ such that $j \leq p_j \leq j(j+1)$ for all $1 \leq j \leq \ell$. 
\end{lemma}
\begin{proof}
The proof goes by induction on $\ell$ and shows that it is possible to pick $p_\ell$ such that the rows of $A(\{p_j\}_{j=1}^\ell)$ are linearly independent. The base case is trivial. In the inductive step, assume $p_1, \cdots, p_{\ell-1}$ have been picked so that the $(\ell-1) \times (\ell-1)$ matrix $A(\{p_j\}_{j=1}^{\ell-1})$ is invertible. We show that for some choice of integer $p_{\ell} \in [p_{\ell-1}, \ell(\ell+1)]$ there does not exist a vector $\bz$ such that $\sum_{i \leq \ell} z_i A_{i,j} = 0$ for all  $j \leq \ell$ where $A = A(\{p_j\}_{j=1}^{\ell})$. We write the first $\ell - 1$ entries of row $A_{\ell}$ as a linear combination of the other $\ell - 1$ rows:
$$   \sum_{i<\ell} z_i A_{i,j} = A_{\ell,j} \;\;\; \forall j < \ell.$$
Since $A(\{p_j\}_{j=1}^{\ell-1})$ is invertible by the inductive hypothesis, there exists a unique solution $\bz^{\star}$ to the above system of linear equations. It remains to show that $ \sum_{i<\ell} z_i^{\star} A_{i,{\ell}} \neq A_{\ell,{\ell}}$, which by the uniqueness of $\bz^{\star}$ implies that there does not exist a vector ${z}$ such that $\sum_{i \leq \ell} z_i A_{i,j} = 0$  for all  $j \leq \ell$. Observe that $A_{\ell,{\ell}} + \sum_{i<\ell} z_i^{\star} A_{i,{\ell}} = (p_{\ell}^{\ell} - (p_{\ell}-1)^{\ell} + \sum_{i<\ell} z_i^{\star}(p_{\ell}^i - (p_{\ell}-1)^i) p_{\ell}^{\ell - i}) / p_{\ell}^{\ell -1}$ and that 
$$p_{\ell}^{\ell} - (p_{\ell}-1)^{\ell} + \sum_{i<\ell} z_i^{\star}(p_{\ell}^i - (p_{\ell}-1)^i) p_{\ell}^{\ell - i}$$
 is a non-zero polynomial of degree $\ell$ that has at most $\ell$ roots. Therefore, there exists $p_{\ell}$ such that $p_{\ell-1} < p_{\ell}  \leq p_{\ell-1} +{\ell} + 1$ and $ \sum_{i<\ell} z_i^{\star} A_{i,{\ell}} \neq A_{\ell,{\ell}}$. So the rows of $A(\{p_j\}_{j=1}^{\ell})$ are linearly independent and the matrix is invertible. We get the bounds on $p_{\ell}$ by the induction hypothesis, $p_{\ell} \leq p_{\ell-1} +{\ell} + 1 \leq (\ell-1)\ell +\ell + 1 \leq \ell(\ell+1)$.
\end{proof}

 We need the following lemma to show the identical on small samples property.

\begin{lemma}
\label{lem:skewed-concentration}
Let $T$ be a uniformly random set of size $|T|$ and consider a set $S$ such that $|T|  \cdot |S| /n  \leq n^{- \epsilon}$ for some constant $\epsilon > 0$, then   $\prob{|S \cap T| \geq \ell} = n^{-\Omega\left(\ell\right)}$.
\end{lemma}
\begin{proof}
 We start by considering a subset $L$ of $S$ of size $ \ell$. We first bound the probability that $L$ is a subset of $T$,
\begin{align*}
\prob{L \subseteq T} \leq \prod_{e \in L} \prob{e \in T} \leq \prod_{e \in L} \frac{|T|}{n} = \left(\frac{|T|}{n}\right)^{\ell}.
\end{align*} We then bound the probability that $|S \cap T| > \ell$ with a union bound over the events that a set $L$ is a subset of $T$, for all subsets $L$ of $S$ of size $\ell$:
\begin{align*}
\prob{|S \cap T| >  \ell}  \leq \sum_{L \subseteq S \given |L| = \ell} \prob{L \subseteq S} 
  \leq \binom{|S|}{\ell} \cdot \left(\frac{|T|}{n}\right)^{  \ell} 
  \leq \left(\frac{|T| \cdot |S|}{n} \right)^{\ell} 
  &\leq n^{ -\epsilon \ell}
\end{align*}
where the last inequality follows from the assumption that $|T| \cdot |S|/n \leq n^{ - \epsilon}$.
\end{proof}

For coverage functions, we let $\ell = \log \log n$.
\begin{lemma}
\label{l:issp}
The identical on small samples property holds for $t = n^{3/5 + \epsilon/2}$.
\end{lemma}
\begin{proof}
Lemma~\ref{lem:skewed-concentration} implies that $|S \cap T_j| \leq \ell = \log \log n$ w.p. $1- \omega(1)$ over $P \sim \UC(\PC)$ for all $j$ for a set $S$ of size at most $n^{3/5 + \epsilon/2}$. Thus, $g(S \cap T_j) = b^T(S \cap T_j)$  for all $j$ w.p. $1- \omega(1)$ by the system of linear equations, which implies the identical on small samples property for $t = n^{3/5 + \epsilon/2}$.
\end{proof}

The gap and curvature properties require bounding the coefficients $\bx$ (Lemma~\ref{l:boundCoeff}). We recall two basic results from linear algebra (Theorems~\ref{t:cramer} and \ref{t:hadamard}) that are used to bound the coefficients.

\begin{theorem}[Cramer's rule]
\label{t:cramer}
Let $A$ be an invertible matrix. The solution to the linear system $Ax = y$ is given by $x_i = \frac{\det A_i}{\det{A}}$,
where $A_i$ is the matrix $A$ with the $i$-th column replaced by the vector $y$.
\end{theorem}

\begin{theorem}[Hadamard's inequality]
\label{t:hadamard}
$\det A \leq \prod \Vert v_i\Vert$, where $\Vert v_i\Vert$ denotes the Euclidean norm of the $i$-th column of $A$.
\end{theorem}

\begin{lemma}
\label{l:boundCoeff}
Let ${\bx^{\star}}$ be the solution to the system of linear equations $\left(A(\{p_j\}_{j=1}^{\ell})\right)  \bx = \by$, then the entries of this solution are bounded: $|x_i^{\star}| \leq \ell^{O(\ell^4)}$.
\end{lemma}
\begin{proof}
Denote ${ A} := A(\{p_j\}_{j=1}^{\ell})$.
By Lemma \ref{lem:invertible}, ${ A}$ is invertible, so let ${\bx^{\star}} = \left({ A}\right)^{-1} \by$.
By Cramer's rule (Theorem~\ref{t:cramer}), $x_i^{\star} = \frac{\det { A}_i}{\det{ A}}$,
where ${ A}_i$ is ${ A}$ with the $i$-th column replaced by the vector $\by$.
Using the bound from Lemma \ref{lem:invertible}, every entry in ${ A}$ can be represented as a rational number, 
with numerator and denominator bounded by $\ell^{O(\ell)}$.
We can multiply by all the denominators, and get an integer matrix with positive entries bounded by $\ell^{O(\ell^3)}$.
Now, by Hadamard's inequality (Theorem~\ref{t:hadamard}), the determinants of the integral ${ A}$ and all the ${ A}_i$'s are integers bounded by $\ell^{O(\ell^4)}$.
Therefore every entry in ${\bx^{\star}}$ can be written as a rational number with numerator and denominator bounded by $\ell^{O(\ell^4)}$.
\end{proof}

Using the bounds previously shown for $\bx^{\star}$, the two following lemmas establish the  gap $\alpha$ and curvature $\beta$ of the good and bad functions $g(\cdot)$ and $b(\cdot)$.

\begin{lemma}
\label{l:gap}
The gap between the good and the bad functions $g(\cdot)$ and $b(\cdot)$  is at least $\alpha = n^{1/5-\epsilon}$ for general coverage functions and at least $\alpha = 2^{\Omega(\sqrt{\log n})}$ for polynomial-size coverage functions.
\end{lemma}
\begin{proof}
We   show the gap between the good and the bad function on a set $S$ of size $k$. Recall that $b(S) \leq r \cdot b^T(k) = r\cdot \sum_{j \given x_j^{\star} > 0, j \leq \ell} x_j^{\star} C_{p_j}(k)$. We can bound each summand as:

\begin{align*}
x_j^{\star} C_{p_j}(k)  & \leq x_j^{\star} p_j && \text{($C_{ p}$ and $c_{p}$ upper bounded by $p$)} \\
& \leq x_j^{\star} \ell(\ell+1) && \text{(Lemma \ref{lem:invertible})} \\
& \leq \ell^{O(\ell^4)} && \text{(Lemma \ref{l:boundCoeff}),}
\end{align*}
and therefore $b^T(k) \leq \ell^{O(\ell^4)}$.
On the other hand, the good function is bounded from below by the cardinality: $g(k) \geq k$.
Plugging in $k = n^{2/5-\epsilon}$, $r = n^{1/5 - \epsilon}$ and $\ell =  \log\log n$, we get the following gap $\alpha$,
$$\frac{g(k)}{b(S)} \geq \frac{n^{2/5}} {n^{1/5}(\log \log n)^{\log ^4   \log n}} \gg n^{1/5-\epsilon}.$$
With $k= 2^{\sqrt{\log n}} , r= 2^{\sqrt{\log n}/2}$, and $\ell = \log \log n$, we get 
$$ \frac{g(k)}{b(S)} \geq \frac{2^{\sqrt{\log n}}} {2^{\sqrt{\log n}/2}(\log \log n)^{\log ^4 \log n}}  = 2^{(1-o(1))\sqrt{\log n}/2}.$$ \end{proof}

\begin{lemma}
\label{l:curv}
The curvature for both the general and polynomial-size  good function is $\beta = o(1)$.
\end{lemma}
 \begin{proof}
Note that $k/r \geq 2^{\sqrt{\log n}/2}$ for both the general and polynomial size cases. Thus, the curvature $\beta$  is 
\begin{align*}
1- \frac{  g(k) }{ r \cdot g(k/r)} & \leq 1 - \frac{ k}{ r \cdot k/r + r \cdot (\log \log n)^{\log ^4 \log n}} \\
& \leq 1 - (1 +  2^{-\sqrt{\log n}/2} \cdot (\log \log n)^{\log ^4 \log n} )^{-1} \\
& = o(1) 
\end{align*}
where the first inequality follows a similar reasoning as the one used to upper bound $b(S)$ in Lemma~\ref{l:gap}. 
\end{proof}

Finally, combining Lemmas~\ref{l:issp}, \ref{l:gap}, and \ref{l:curv}, we get Lemma~\ref{l:abconstruction}.

\subsubsection*{Constructing the Masking Function}

To obtain the desired properties of the masking functions $m^+$ and masking fraction $m^+$, each child $a_i$ in the universe of $g+b$ is divided into $n^{3/5}$ children $a_{i,1}, \ldots a_{i,n^{3/5}}$ of equal weights $w(a_i) / n^{3/5}$. For each masking element, draw $j \in \UC([n^{3/5}])$, then this masking element covers $a_{i,j}$ for all $i$. The function $m^+(S)$ is then the total weight covered by masking elements $S$ and  the masking fraction $m(S)$ is the fraction of $j \in [n^{3/5}]$ such that $j$ is drawn for at least one element in $S$. Lemmas \ref{l:masking} and \ref{l:gapmasking} show the masking property on large samples and the $\alpha$-gap for masking elements. We begin by stating the Chernoff bound, used in Lemma~ \ref{l:masking}.

\begin{lemma}[Chernoff Bound]
\label{l:chernoff}
Let $X_1, \dots, X_n$ be independent indicator random variables such that $\prob{X_i = 1} = 1$. Let $X = \sum_{i=1}^n X_i$ and $\mu = \expect{X}$. For $0 < \delta < 1$,
$$\prob{|X - \mu| \geq \delta \mu} \leq 2 e^{- \mu \delta^2 / 3}.$$
\end{lemma}

\begin{lemma}
\label{l:masking}
Consider the construction above for masking elements, then the masking on large samples property holds for $t = n^{3/5 + \epsilon}$.
\end{lemma}
\begin{proof}
First, we show that a large set $S$ contains a large number of masking elements with exponentially high probability.\footnote{Formally, with exponentially high probability means with probability at least $1 - e^{-\Omega(n^{\epsilon})}$ for some constant $\epsilon > 0$.} We then show that a large number of masking elements covers all the children with exponentially high probability, thus $m(S \cap M) = 1$.

Masking elements are a $1 - o(1)$ fraction of $N$ since there are $n - rk= n - n^{1/5 - \epsilon}  n^{2/5 - \epsilon}$ masking elements. By Chernoff bound (Lemma~\ref{l:chernoff}), a set of size at  least $n^{3/5 + \epsilon} $ contains at least $n^{3/5 + \epsilon/2}  $ masking elements with exponentially high probability. By another Chernoff bound, with $n^{3/5 + \epsilon/2}  $ masking elements, at least one of these elements cover a fixed child $a_{i,j}$ with exponentially high probability. By a union bound, this holds for all $j \in [n^{3/5}]$. Finally, note that if a set of masking elements cover $a_{i,1}, \ldots, a_{i,n^{3/5}}$ for some $i$, this set covers $a_{i,1}, \ldots, a_{i,n^{3/5}}$ for all $i$. Thus w.p. at least $1 - n^{-\omega(1)}$, $m(S \cap M) = 1$.
\end{proof}

\begin{lemma}
\label{l:gapmasking}
The masking function $m$ has a gap $\alpha = n^{1/5 + \epsilon}$ with the good function $g$.
\end{lemma} 
\begin{proof}
We first bound the value of all good and bad elements, and then bound the fraction of that total value obtained by $k$ masking elements. The value of all bad elements is
\begin{align*}
  \sum_{j=1,j \neq j}^m b^T(T_j)  &=  (r-1) \sum_{1\leq j \leq \ell \given x_j^{\star} \geq 0} x_j^{\star}C_{p_j}(k) && \text{(Definition of $b^T$)} \\
& \leq    r   \sum_{1\leq j \leq \ell} \ell^{O(\ell^4)} C_{p_j}(k)  && \text{(Lemma \ref{l:boundCoeff})} \\
& \leq   r  \sum_{1\leq j \leq \ell} \ell^{O(\ell^4)} p_j  &&\text{($C_{p}, c_{p} \leq p$)} \\
& \leq  r \sum_{1\leq j \leq \ell} \ell^{O(\ell^4)}   &&\text{(Lemma \ref{lem:invertible})} \\
& \leq   o(k)   && (\ell = \log \log n, r  = n^{1/5 - \epsilon}, k=n^{2/5 - \epsilon}) \\
& \leq    o(g(k))  && 
\end{align*}
Now note that a masking element covers a $1/ n^{3/5}$ fraction of the value of all good and bad elements by the above construction. Thus, $k = n^{2/5 - \epsilon}$ masking elements cover at most a $1/ n^{1/5 + \epsilon}$ fraction of the total value of all good and bad elements, combining this with the total value of bad elements that is upper bounded by $o(g(k))$ concludes the proof. 
\end{proof}

Combining Lemmas~\ref{l:abconstruction}, \ref{l:masking}, and \ref{l:gapmasking}, we obtain an $(n^{1/5 - \epsilon}, o(1))$-gap. The main result for exponential size coverage functions then follows from Theorem~\ref{t:gap}.

\cexpo*

\subsubsection*{From Exponential to Polynomial Size Coverage Functions}

We modify $C_{p}$ to use primitives $c_{p}$ which are coverage with \emph{polynomially} many children.  The function class $\mathcal{F}(g,b,m,m^+)$ obtained are then coverage functions over a polynomial-size universe. The matrix $A$ for polynomial size coverage functions is identical as in the general case.  We lower the cardinality constraint to $k = 2^{\sqrt{\log n}} = |T_j|$ so that the functions $c_{p}(S \cap T_j)$ need to be defined over only $2^{\sqrt{\log n}}$ elements. We also lower the number of parts to $r = 2^{\sqrt{\log n}/2} $.

The main technical challenge is to obtain symmetric coverage functions for sets of size at most $\ell$ with polynomially many children. We start by reducing the problem to constructing such functions with certain properties in Lemma~\ref{l:zeta}. We then construct such functions and prove they satisfy these properties in Lemma~\ref{l:symm}. Combining these Lemmas, we obtain a $(2^{\Omega(\sqrt{\log n})}, o(1))$-gap (Lemma~\ref{l:poly}).

\lzeta*
\begin{proof}
The proof starts from the construction for $C_{p}$ with exponentially many children over a ground set of size $p$ and modifies it into a coverage function with polynomially many children while satisfying the desired conditions.   For each  $z \leq k$, replace all children in $C_{p}$ that are covered by exactly $z$ elements with  $\zeta^z(\cdot)$. Define $C_{p}^z$ to be $C_{p}$ but only with these children that are covered by exactly $z$ elements.  Let the new children from $\zeta^z(\cdot)$ be such that $\zeta^z(k) = C_{p}^z(k)$.

Clearly $c_{p}$ has polynomially many children in $n$ since each $\zeta^z(\cdot)$ has polynomially many children. Then, note that  
$$c_{p}(k) = \sum_{z=1}^{k} \zeta^z(k) = \sum_{z=1}^{k} C_{p}^z(k) =  C_{p}(k).$$
 Finally, we show that $c_{p}(S) = C_{p}(y)$ for all $S$ such that $|S| = y \leq l$. Note that it suffices to show that $\zeta^z(S) =  C_{p}^z(y)$ for all $z$, which we prove by induction on $y$. The base case $y=0$ is trivial. If $y > 0$, then consider some set $S$ such that $|S| = y$ and let $e \in S$. By the inductive hypothesis, $\zeta^z(S \setminus e) =  C_{p}^z(y-1)$.   Let $T$ be the set of all children in $\zeta^z(\cdot)$ not covered by $S \setminus e$. Define $\zeta^z_T(\cdot)$ to be $\zeta^z(\cdot)$ but only with children in $T$. Since all children in $T$ are covered by $ z$ distinct elements that are not in $S \setminus e$,
$$\sum_{e' \not \in S \setminus e} \zeta^z_T(e') = z \cdot (\zeta^z(k) - \zeta^z(y-1)).$$
 By the assumptions on  $\zeta^z$, $\zeta^z(S)= \zeta^z(S \setminus e \cup e')$. For any $e' \not \in S \setminus e$, by combining with both $\zeta^z(S) = \zeta^z(S \setminus e) + \zeta^z_T(e)$ and  $\zeta^z(S \setminus e \cup e') = \zeta^z(S \setminus e) + \zeta^z_T(e')$, 
 $$\zeta^z_T(e)=  \zeta^z_T(e') = z \cdot (\zeta^z(k) - \zeta^z(y-1)) / (k - y + 1).$$
  So, 
\begin{align*}
\zeta^z(S) & = \zeta^z(S \setminus e) + \zeta^z_T(e) \\
& = \zeta^z(y-1) + z \cdot (\zeta^z(k) - \zeta^z(y-1)) / (k - y + 1) \\
& = C_{p}^z(y-1) + z \cdot (C_{p}^z(k) - C_{p}^z(y-1)) / (k - y + 1) \\
& = C^z_{p}(y)
\end{align*}
where the last equality is obtained for $C^z_{p}$ similarly as how it was obtained for $\zeta^z$.
\end{proof}

We now construct such  $\zeta^z$. Assume without loss that $k$ is prime (o.w. pick some prime close to $k$).  Given $\ba \in [k]^{\ell}$, and $x \in [z]$,  let $h_{\ba}(x) := \sum_{i \in [\ell]} a_i x^i \mod k$. The children in $\zeta^z$ are $U = \{ \ba \in [k]^{\ell} \given h_{\ba}(x_1) \neq h_{\ba}(x_2) \text{ for all distinct } x_1, x_2 \in [z]\}$. The $k$ elements are $\{j \given 0 \leq j < k\}$. Child $\ba$ is covered by elements $\{h_{\ba}(x) \given x \in [z]\}$. Note that $|U| \leq k^{\ell} = 2^{\ell \sqrt{\log n}}$ and we pick $\ell = \log \log n$ as previously. The following lemma is useful to show the symmetricity of $\zeta^z$.

\begin{lemma}
\label{l:indep} Let $\ba$ be a uniformly random child, then $\prob{h_{\ba}(x_1) = j_1, \ldots, h_{\ba}(x_{\ell}) = j_{\ell}}$ is independent of distinct $x_1, \ldots, x_{\ell} \in [z]$ and $j_1, \ldots, j_{\ell} \in [p]$. More precisely, $$\prob{h_{\ba}(x_1) = j_1, \ldots, h_{\ba}(x_{\ell}) = j_{\ell}} = \prod_{i=1}^{\ell} \frac{1}{p + 1 -i}.$$
\end{lemma}
\begin{proof}
It is well-known that the random variables $(h_{\ba}(0), \ldots, h_{\ba}(z -1))$ where $\ba$ is chosen uniformly at random from $[p]^{\ell}$ are $\ell$-wise independent since  $h_{\ba}(\cdot)$ is a polynomial of degree $\ell - 1$, so 
$$ \textnormal{Pr}_{\ba \sim \mathcal{U}([p]^{l})}(h_{\ba}(x_1) = j_1, \ldots, h_{\ba}(x_{\ell}) = j_{\ell}) = \prod_{i=1}^{\ell} \prob{ h_{\ba}(x_i) = j_i }.
$$
  By throwing away all children $\ba$ such that there exists distinct $x_1$, $x_2$ with  $h_{\ba}(x_1) = h_{\ba}(x_2)$, we obtain the following by combining with the symmetry of the children $\ba$ removed (there exists exactly one polynomial defined by some $\ba$ passing through any collection of $\ell$ points):
$$ \prob{h_{\ba}(x_1) = j_1, \ldots, h_{\ba}(x_{\ell}) = j_{\ell}} = \prod_{i=1}^{\ell} \prob{ h_{\ba}(x_i) = j_i | h_{\ba}(x_i) \not \in \{j_1, \ldots, j_{i-1}\}}. $$
Finally, note that 
$$ \prod_{i=1}^{\ell} \prob{ h_{\ba}(x_i) = j_i | h_{\ba}(x_i) \not \in \{j_1, \ldots, j_{i-1}\}} = \prod_{i=1}^{\ell} \frac{1}{p + 1 -i}.
$$
by the symmetry induced by $a_0$.
\end{proof}
We are now ready to show the main lemma for the coverage functions $\zeta^z(\cdot)$ .
\lsymm*
\begin{proof}
Let $\ba$ be a child chosen uniformly at random and $S$ be a set of size at most $\ell$. Then, $\zeta^z(S)  = k \cdot  \prob{\cup_{j \in S} (\exists x \in [z] \text{ s.t. } h_{\ba}(x)=j)}$ and
$$
 \prob{\cup_{j \in S} (\exists x \in [z] \text{ s.t. } h_{\ba}(x)=j)}  = \sum_{T \subseteq S} (-1)^{|T| + 1} \prob{T \subseteq \{h_{\ba}(x) \given x \in [z]\}} $$ 
 by inclusion-exclusion. Note that $\prob{T \subseteq \{h_{\ba}(x) \given x \in [z]\}}$ only depends on the size of $T$ by Lemma \ref{l:indep}. Therefore $\zeta^z(S)$ only depends on the size of $S$ and $\zeta^z(\cdot)$ is symmetric for all sets of size at most $\ell$.
\end{proof}

We obtain an $(\alpha = 2^{\Omega(\sqrt{\log n})}, \beta = o(1))$-gap for polynomial sized coverage functions by using the primitives $c_p$.

\lpoly*
\begin{proof}
We construct $g,b,m,m^+$ as in the general case but in terms of primitives $c_{p}$ instead of $C_{p}$. By Lemmas~\ref{l:zeta} and \ref{l:symm}, we obtain the same matrix $A$ and coefficients $\bx^{\star}$ as in the general case, so the identical on small samples property holds. The masking on large samples holds identically as for general coverage functions. The gap and curvature properties are shown in Lemmas~\ref{l:gap} and \ref{l:curv}.
\end{proof}

We conclude with the main result for coverage functions by combining Claim~\ref{c:expo}, Lemma~\ref{l:poly}, and Theorem~\ref{t:gap}.
\tcover*

\newpage
\section{Algorithms for OPS}
\label{s:mpalgorithms}
\subsection*{OPS via Estimates of Expected Marginal Contributions}
We denote by $\mathcal{S}_i$ and $\mathcal{S}_{-i}$ the collections of all samples containing and not containing element $e_i$ respectively. The estimate $\hat{v}_i$ is then the difference in the average value of a sample in $\mathcal{S}_i$ and  the average value of a sample in $\mathcal{S}_{-i}$. By standard concentration bounds (Hoeffding's inequality, Lemma~\ref{l:hoeffding}), these are good estimates of $\expect[S \sim \DC|e_i \not \in S]{f_S(e_i)}$ for product distributions $\DC$ (Lemma~\ref{l:est}).

\begin{lemma}[Hoeffding's inequality]
\label{l:hoeffding}
Let $X_1, \dots, X_n$ be independent random variables with values in $[0,b]$. Let $X = \frac{1}{m}\sum_{i=1}^m X_i$ and $\mu = \expect{X}$. Then for every $0 < \epsilon < 1$,
$$\prob{|\bar{X} - \expect{\bar{X}}| \geq \epsilon} \leq 2 e^{- 2 m \epsilon^2 /  b^2}.$$
\end{lemma}

\lest*
\begin{proof}
Let $\epsilon \geq f(N) / n^{c}$ for some constant $c$. Since $\mathcal{D}$ is a product distribution with marginals bounded away from $0$ and $1$,  there are at least $2n^{2c + 1}$ samples containing element $e_i$ and at least $2n^{2c + 1}$ samples  not containing $e_i$  for all $i$, with exponentially high probability, with a sufficiently large polynomial number of samples.  Then by Hoeffding's inequality (Lemma~\ref{l:hoeffding} with $m = 2n^{2c + 1}$ and $b = f(N)$), 
$$\textnormal{Pr}\left(\left|\frac{1}{|\mathcal{S}_i|}\sum_{S \in \mathcal{S}_i} f(S) -\expect[S \sim \DC|e_i \in S]{f(S)}\right| \geq \epsilon/2 \right) \leq 2 e^{- 4 n^{2c + 1} (\epsilon/2)^2 / f(N)^2} \leq 2e^{- n^{2c +1}  / n^{2c}} \leq 2e^{-n} $$
and similarly,
$$\textnormal{Pr}\left(\left|\frac{1}{|\mathcal{S}_{-i}|}\sum_{S \in \mathcal{S}_{-i}} f(S)- \expect[S \sim \DC | e_i \not \in S]{f(S)}\right|  \geq \epsilon/2 \right)\leq 2e^{-n} .$$
 Since $$\hat{v}_i = \frac{1}{|\mathcal{S}_i|}\sum_{S \in \mathcal{S}_i} f(S) - \frac{1}{|\mathcal{S}_{-i}|}\sum_{S \in \mathcal{S}_{-i}} f(S)$$ and $$\expect[S \sim \DC|e_i \not \in S]{f_S(e_i)} = \expect[S \sim \DC|e_i \not \in S]{f(S \cup e_i)} - \expect[S \sim \DC|e_i \not \in S]{f(S)} = \expect[S \sim \DC|e_i  \in S]{f(S)} - \expect[S \sim \DC|e_i \not \in S]{f(S)}$$where the second equality is since $\DC$ is a product distribution, the claim then holds  with probability at least $1 - 4 e^{-n}$. \end{proof}

\subsection*{A Tight Approximation for Submodular Functions}

Let $\mathcal{D}_i$ be the uniform distribution over all sets of size $i$. Define the distribution $\mathcal{D}^{sub}$ to be the distribution which draws from $\mathcal{D}_k$, $\mathcal{D}_{\sqrt{n}}$, and $\mathcal{D}_{\sqrt{n}+1}$ at random.  This section is devoted to show that Algorithm~\ref{alg:ubsub} is an $\tilde{\Omega}(n^{-1/4})$-\texttt{OPS} algorithm over $\mathcal{D}^{sub}$ (Theorem~\ref{t:ubsubmodular}). Define $\mathcal{S}_{i,j}$ and $\mathcal{S}_{-i,j}$ to be the collections of samples of size $j$ containing and not containing $e_i$ respectively. For Algorithm~\ref{alg:ubsub}, we use a slight variation of \textsc{EEMC}($\mathcal{S}$) where the estimates are $$\hat{v}_i = \frac{1}{|\mathcal{S}_{i,\sqrt{n}+1}|} \sum_{S \in \mathcal{S}_{i,\sqrt{n}+1}}f(S) - \frac{1}{|\mathcal{S}_{-i,\sqrt{n}}|} \sum_{S \in \mathcal{S}_{-i,\sqrt{n}}}f(S).$$
These are good estimates of $\expect[S \sim \DC_{\sqrt{n}}|e_i\not \in S]{f_S(e_i)}$, as shown by the following lemma. The proof follows almost identically as the proof for Lemma~\ref{l:est}.
 \begin{lemma}
\label{l:est2}
With probability at least $1 - O(e^{-n})$, the estimations $\hat{v}_i$ defined above are $\epsilon$ accurate, for any $\epsilon \geq f(N) / \poly(n)$ and for all $e_i$,  i.e., 
$$ | \hat{v}_i - \expect[S \sim \DC_{\sqrt{n}}|e_i\not \in S]{f_S(e_i)} | \leq \epsilon.$$
\end{lemma}
\begin{proof}
Let $\epsilon \geq f(N) / n^{c}$ for some constant $c$. With a sufficiently large polynomial number of samples,  $\mathcal{S}_{i,\sqrt{n}+1}$ and  $\mathcal{S}_{-i,\sqrt{n}}$ are of size at least $2n^{2c + 1}$ with exponentially high probability.  Then by Hoeffding's inequality (Lemma~\ref{l:hoeffding} with $m = 2n^{2c + 1}$ and $b = f(N)$), 
$$\textnormal{Pr}\left(\left|\frac{1}{|\mathcal{S}_{i,\sqrt{n}+1}|}\sum_{S \in \mathcal{S}_{i,\sqrt{n}+1}} f(S) -\expect[S \sim \DC_{\sqrt{n}+1}|e_i \in S]{f(S)}\right| \geq \epsilon/2 \right) \leq 2 e^{- 4 n^{2c + 1} (\epsilon/2)^2 / f(N)^2}  \leq 2e^{-n} $$
and similarly,
$$\textnormal{Pr}\left(\left|\frac{1}{|\mathcal{S}_{-i,\sqrt{n}}|}\sum_{S \in \mathcal{S}_{-i,\sqrt{n}}} f(S)- \expect[S \sim \DC_{\sqrt{n}}|e_i \not \in S]{f(S)}\right|  \geq \epsilon/2 \right)\leq 2e^{-n} .$$
By the definition of $\hat{v}_i$ and since $$\expect[S \sim \DC_{\sqrt{n}}|e_i \not \in S]{f_S(e_i)}  = \expect[S \sim \DC_{\sqrt{n}+1}|e_i  \in S]{f(S)} - \expect[S \sim \DC_{\sqrt{n}}|e_i \not \in S]{f(S)},$$ the claim then holds  with probability at least $1 - 4 e^{-n}$. \end{proof}

Next, we show a simple lemma that is useful when returning random sets (Lemma~\ref{l:k/n}). The analysis is then divided into two cases, depending if a random set $S \sim \DC_{\sqrt{n} }$ has low value or not. If a set has low value, then we obtain an $\tilde{\Omega}(n^{-1/4})$-approximation by Corollary~\ref{c:1}. Corollary~\ref{c:1} combines Lemmas \ref{l:t/kn1/2} and \ref{l:k/t} which respectively obtain $t  / (4k\sqrt{n})$ and $k/t$ approximations. If a random set has high value,  then we obtain an $n^{-1/4}$-approximation by Corollary~\ref{c:2}. Corollary~\ref{c:2} combines Lemmas \ref{l:kn1/2} and \ref{l:1/k} which respectively obtain $k / (4\sqrt{n})$ and $1/k$ approximations.
 
  \begin{lemma}
  \label{l:k/n}
For any monotone submodular function $f(\cdot)$,  the value of a uniform random set $S$ of size $k$ is a $k/n$-approximation to $f(N)$.
\end{lemma}
\begin{proof}
Partition the ground set into sets of size $k$ uniformly at random. A uniform random set of this partition is a $k/n$-approximation to  $f(N)$ in expectation by submodularity. A uniform random set of this partition is also a uniform random set of size $k$.
\end{proof}
In the first case of the analysis, we assume that $\expect[S \sim \DC_{\sqrt{n} }]{f(S)} \leq f(S^{\star}) / 4$. Let $j'$ be the largest $j$ such that bin $B_{j}$ contains at least one element $e_i$ such that $\hat{v}_i  \geq f(S^{\star}) / (2k)$. So any element $e_i \in B_j$, $j \leq j'$ is such that $\hat{v}_i  \geq f(S^{\star}) / (4k)$. Define $B^{\star} = \argmax_{B_j: j \leq j'} f(S^{\star} \cap B_j)$ to be the bin $B$  with high marginal contributions that has the highest value from the optimal solution. Let $t$ be the size of $B^{\star}$.

\begin{lemma}
\label{l:t/kn1/2}
If $\expect[S \sim \DC_{\sqrt{n}}]{f(S)} \leq f(S^{\star}) / 4$, then a uniformly random subset of bin $B^{\star}$ of size $\min\{k, t\}$  is a $(1-o(1))\cdot\min( 1/4, t  / (4k\sqrt{n}))$-approximation to $f(S^{\star})$.
\end{lemma}
\begin{proof}
Note that
\begin{align*}
\expect[S\sim \DC_{t / \sqrt{n}}|S \subseteq B^{\star}]{f(S)} & \geq \expect[S \sim \DC_{\sqrt{n}}]{f(S \cap B^{\star})} && \text{(submodularity)}\\
 &\geq \expect[S \sim \DC_{\sqrt{n}}]{\sum_{e_i \in S \cap B^{\star}} f_{(S \cap B^{\star}) \setminus e_i}(e_i)} && \text{(submodularity)}\\
&\geq \expect[S \sim \DC_{\sqrt{n}}]{\sum_{e_i \in S \cap B^{\star}} \expect[S' \sim \DC_{\sqrt{n}}|e_i \not \in S']{f_{S'}(e_i)} } && \text{(submodularity)}\\
&\geq \expect[S \sim \DC_{\sqrt{n}}]{\sum_{e_i \in S \cap B^{\star}} (\hat{v}_i - f(N) / n^3)} && \text{(Lemma~\ref{l:est})}\\
&= \expect[S \sim \DC_{\sqrt{n}}]{| S \cap B^{\star}|} (1- o(1))\frac{f(S^{\star})}{4k} && (\hat{v}_i  \geq f(S^{\star}) / (4k) \text{ for } e_i \in B^{\star}) \\
&= (1-o(1))\frac{t  f(S^{\star})}{4k\sqrt{n}} 
\end{align*} 

 If $t  / \sqrt{n}\geq k $, then a uniformly random subset of bin $B^{\star}$ of size $k$ is a $k \sqrt{n} / t$ approximation to $\expect[S\sim \DC_{t / \sqrt{n}}|S \subseteq B^{\star}]{f(S)}$ by Lemma~\ref{l:k/n}, so a $(1 - o(1))/ 4$ approximation to $f(S^{\star})$ by the above inequalities. Otherwise, if $t  / \sqrt{n} < k $, then a uniformly random subset of bin $B^{\star}$ of size $\min\{k, t\}$ has value at least  $\expect[S\sim \DC_{t / \sqrt{n}}|S \subseteq B^{\star}]{f(S)}$ by monotonicity, and is thus a $(1-o(1)) t /(4k\sqrt{n})$ approximation to $f(S^{\star})$.

\end{proof}


\begin{lemma}
\label{l:k/t}
If $\expect[S \sim \DC_{\sqrt{n}}]{f(S)} \leq f(S^{\star}) / 4$, a uniformly random subset of bin $B^{\star}$ of size $\min(k, t)$  is an $\tilde{\Omega}(\min(1, k/t))$-approximation to $f(S^{\star})$.
\end{lemma}
\begin{proof}
We start by bounding the value of optimal elements not in bin $B_j$ with $j \leq j'$.
\begin{align*}
& f(S^{\star} \setminus (\cup_{B_j : j \leq j'} B_j)) \\
  \leq &\expect[S \sim \DC_{\sqrt{n}}]{f(S \cup S^{\star} \setminus (\cup_{B_j : j \leq j'} B_j))} && \text{(monotonicity)}\\ 
  \leq &\expect[S \sim \DC_{\sqrt{n}}]{f(S) + \sum_{e_i \in  (S^{\star} \setminus (\cup_{B_j : j \leq j'} B_j)) \setminus S} \expect[S' \sim \DC_{\sqrt{n}}|e_i \not \in S']{f_{S'}(e_i)} } && \text{(submodularity)}\\
\leq &f(S^{\star})/ 4 + \expect[S \sim \DC_{\sqrt{n}}]{\sum_{e_i \in  (S^{\star} \setminus (\cup_{B_j : j \leq j'} B_j)) \setminus S} (\hat{v}_i + f(N) / n^3)} && \text{(assumption and Lemma~\ref{l:est})}\\
 \leq &f(S^{\star})/ 4 + f(S^{\star})/ 2 + f(S^{\star})/ n&& \text{(definition of } j')
\end{align*}
Since $f(S^{\star}) \leq f(S^{\star} \cap (\cup_{B_j : j \leq j'} B_j)) + f(S^{\star} \setminus (\cup_{B_j : j \leq j'} B_j))$ by submodularity, we get that $f(S^{\star} \cap (\cup_{B_j : j \leq j'} B_j))  \geq f(S^{\star}) / 5$. Since there are $3 \log n$ bins, $f(S^{\star} \cap B^{\star})$ is a $ 3 \log n$ to $f(S^{\star}) / 5$ by submodularity and the definition of $B^{\star}$. By monotonicity, $f(B^{\star})$ is a $ 15 \log n$ to $f(S^{\star})$.   Thus, by Lemma~\ref{l:k/n},  a random subset $S$ of size $k$ of bin $B^{\star}$ is an $\tilde{\Omega}(\min(1, k/t))$ approximation to $f(S^{\star})$.
 
\end{proof}

\begin{corollary}
\label{c:1}
 If $\expect[S \sim \DC_{\sqrt{n}}]{f(S)} \leq f(S^{\star}) / 4$,  a uniformly random subset of size $\min\{k, |B_j|\}$ of a random bin $B_j$ is an $\tilde{\Omega}(n^{-1/4})$-approximation to $f(S^{\star})$.
\end{corollary}
\begin{proof}
With probability $1/ (3 \log n)$, the random bin is $B^{\star}$ and we assume this is the case. By Lemma~\ref{l:t/kn1/2} and \ref{l:k/t}, a random subset of $B^{\star}$ of size $\min(k, t)$ is both an $\Omega(\min( 1, t / (k\sqrt{n})))$ and an  $\tilde{\Omega}(\min(1, k/t))$ approximation to $f(S^{\star})$. Assume $ t / (k\sqrt{n}) \leq 1$ and  $k/t \leq 1$, otherwise we are done. Finally, note that if $t/k \geq n^{1/4}$, then $ \Omega(t / (k\sqrt{n})) \geq \Omega(n^{-1/4})$,  otherwise, $\tilde{\Omega}(k/t) \geq \tilde{\Omega}(n^{-1/4})$.
\end{proof}

In the second case of the proof, we assume that $\expect[S \sim \DC_{\sqrt{n} }]{f(S)} \geq f(S^{\star}) / 4$
\begin{lemma}
\label{l:kn1/2}
For any monotone submodular function $f$, if $\expect[S \sim \DC_{\sqrt{n}}]{f(S)} \geq f(S^{\star}) / 4$, then a uniformly random set of size $k$ is a $\min(1/4, k / (4\sqrt{n}))$ approximation to $f(S^{\star})$.
\end{lemma}
\begin{proof}
If $k \geq \sqrt{n}$, then a uniformly random set of size $k$ is a $1/4$-approximation to $f(S^{\star})$ by monotonicity. Otherwise, a uniformly random subset of  size $k$ of $N$ is a uniformly random subset of size $k$ of a uniformly random subset of size $\sqrt{n}$ of $N$. So by Lemma~\ref{l:k/n},
$$\expect[S  \sim \DC_k]{f(S)} \geq \frac{k}{\sqrt{n}} \cdot\expect[S \sim \DC_{\sqrt{n}}]{f(S)} \geq  \frac{k}{4 \sqrt{n}} \cdot  f(S^{\star}).$$
\end{proof}

  \begin{lemma}
  \label{l:1/k}
For any monotone submodular function $f(\cdot)$,  the sample $S$ with the largest value among at least $n \log n$ samples of size $k$ is a $1/k$-approximation to $f(S^{\star})$ with high probability.
\end{lemma}
\begin{proof}
By submodularity, there exists an element $e_i^{\star}$ such that $\{e_i^{\star}\}$ is a $1/k$-approximation to the optimal solution. By monotonicity, any set which contains $e_i^{\star}$ is a $1/k$-approximation to the optimal solution. After observing $n \log n$ samples, the probability of never observing a set that contains $e_i^{\star}$ is polynomially small.
\end{proof}

\begin{corollary}
\label{c:2}
If $\expect[S \sim \DC_{\sqrt{n}}]{f(S)} \geq f(S^{\star}) / 4$, then the sample of size $k$ with the largest value is a $\min(1/4,n^{-1/4}/4)$ approximation to $f(S^{\star})$.
\end{corollary}
\begin{proof}
By  \ref{l:kn1/2} and Lemma~\ref{l:1/k},  the sample of size $k$ with the largest value $k$ is a $\min(1/4, k / (4\sqrt{n}))$ and a $1/k$ approximation to $f(S^{\star})$. If $k \geq n^{1/4}$, then $\min(1/4, k / (4\sqrt{n})) \geq \min(1/4,n^{-1/4}/4)$, otherwise, $1/k \geq 1/ n^{-1/4}$.
\end{proof}

By combining Corollaries \ref{c:1} and \ref{c:2}, we obtain the main result for this section.
\tubsubmodular*

\subsection*{Bounded Curvature and Additive Functions}
The algorithm \textsc{MaxMargCont} simply returns the $k$ elements with the largest estimate $\hat{v}_i$.
\tcurv*
\begin{proof}
Let $S^{\star} = \{e^{\star}_1, \ldots, e^{\star}_k\}$ be the optimal solution and $S = \{e_1, \ldots, e_k\}$ be the set returned by Algorithm~\ref{a:curv}.  Let $S_i^{\star} := \{e^{\star}_1, \ldots, e^{\star}_i\}$ and $S_{i} := \{e_{1}, \ldots, e_{i}\}$. If   $e_j \not \in S$, then
\begin{align*}
f_{S_{i-1}}(e_i)  & \geq (1-c) \cdot \expect[S \sim \DC|e_i \not \in S]{f_S(e_i)} && \text{(curvature)} \\
& \geq (1-c) \hat{v}_{i} - \frac{f(N)}{n^2} && \text{Lemma~\ref{l:est} with $\epsilon = \frac{f(N)}{(1-c)n^2}$} \\
& \geq (1-c) \hat{v}_{j} - \frac{f(N)}{n^2} && ( e_i \in S, e_j \not \in S \text{ and by Algorithm~\ref{a:curv}}) \\
& \geq (1-c)^2 f_T(e_j) - \frac{f(N)}{n^2} && \text{(curvature)} 
\end{align*}
for any set $T$ which does not contain $e_j$ and we conclude that
$$f(S) =  \sum_{i=1}^k f_{S_{i-1}}(e_i) \geq (1-c)^2\left(\sum_{i=1}^k f_{S_{i-1}^{\star}}(e^{\star}_i)\right) - \frac{k f(N)}{n^2} = ((1-c)^2-o(1))f(S^{\star}). $$
\end{proof}

\newpage

\section{Recoverability}
\label{s:mprecoverability}
A function $f$ is recoverable for distribution $\mathcal{D}$ if  given samples drawn from $\mathcal{D}$, it is possible to output a function $\tilde{f}(\cdot)$ such that for all sets $S$,
$\left(1 - 1/n^2\right) f(S) \leq \tilde{f}(S) \leq \left(1 +  1/ n^2 \right) f(S)$
with high probability over the samples and the randomness of the algorithm.
\greedyrecThm*
\begin{proof}
We show that the greedy algorithm with $\tilde{f}(\cdot)$ for a recoverable function performs well. The proof follows similarly as the classical analysis of the greedy algorithm. We start with submodular functions and denote by $S_i = \{e_1, \cdots, e_i\}$ the set obtained after the $i$th iteration. Let $S^{\star}$ be the optimal solution, then by submodularity,
\begin{align*}
f(S^{\star}) & \leq f(S_{i-1}) + \sum_{e \in S^{\star} \setminus S_{i-1}} f_{S_{i-1}}(e) \\
& \leq f(S_{i-1}) + \sum_{e \in S^{\star} \setminus S_{i-1}} \left(\left(\frac{1+1/n^2}{1-1/n^2}\right)f(S_i) - f(S_{i-1})\right)
\end{align*}
where the second inequality follows from $\tilde{f}(S_i) \geq \tilde{f}(S_{i-1} \cup \{e\})$ for all $e \in S^{\star} \setminus S_{i-1}$ by the greedy algorithm, so $(1+ 1/n^2) f(S_i) \geq (1 - 1/n^2) f(S_{i-1} \cup\{e\})$. We therefore get that 
$$f(S^{\star}) \leq (1-k) f(S_{i-1}) + k \left(\frac{1+1/n^2}{1-1/n^2} \right)f(S_i) .$$
By induction and similarly as in the analysis of the greedy algorithm, we then get that 
$$f(S_k) \geq \left(\frac{1-1/n^2}{1+1/n^2}\right)^k \left(1 - (1- 1/k)^k\right)f(S^{\star}).$$

Since 
$$\left(\frac{1-1/n^2}{1+1/n^2}\right)^k \geq \left(1 - \frac{2}{n^2}\right)^k \geq 1 - 2k / n^2 \geq  1 - 2/n$$
and $\left(1 - (1- 1/k)^k\right) \geq 1 - 1/e$, the greedy algorithm achieves an $(1 - 1/e - o(1))$-approximation for submodular functions.

For additive functions, let $S$ be the set returned by the greedy algorithm and $\hat{v}_i = \tilde{f}(\{e_i\})$, then 
$$f(S) = \sum_{e_i \in S}  f(\{e_i\}) \geq  \left(\frac{1}{1+1/n^2}\right)  \sum_{e_i \in S}  \hat{v}_i \geq \left(\frac{1}{1+1/n^2}\right)  \sum_{e_i \in S^{\star}}  \hat{v}_i \geq \left(\frac{1-1/n^2}{1+1/n^2}\right) f(S^{\star}).$$
We therefore get a $(1 - o(1))$-approximation for additive functions.

\end{proof}

\laddrec*
\begin{proof}
We have already shown that the expected marginal contribution of an element to a random set of size $k-1$ can be estimated from samples for submodular functions\footnote{For simplicity, this proof uses estimations that we know how to compute. However, The values $f(\{e_i\})$ can be recovered exactly by solving a system of linear equations where each row corresponds to a sample, provided that the matrix for this system is invertible, which is the case with a sufficiently large number of samples by using results from random matrix theory such as in the survey by \citet{BS06}.}. In the case of additive functions, this marginal contribution of an element is its value $f(\{e_i\})$. 

We apply Lemma~\ref{l:est} with $\epsilon = f(\{e_i\}) / n^2$ to compute $\hat{v}_i$ such that $|\hat{v}_i - f(\{e_i\})| \leq f(\{e_i\}) / n^2$. Note that $\epsilon = f(\{e_i\}) / n^2$ satisfies $\epsilon \geq f(S^{\star}) / \poly(n)$ since $v_{min} \geq  v_{max} / \poly(n)$. Let $\tilde{f}(S) = \sum_{e_i \in S} \hat{v}_i$, then $$\tilde{f}(S) \leq \sum_{e_i \in S} (1 + 1/n^2)f(\{e_i\}) = (1 + 1/n^2)f(S)$$ and  $$\tilde{f}(S) \geq \sum_{e_i \in S} (1 - 1/n^2)f(\{e_i\})  = (1 - 1/n^2)f(S).$$
\end{proof}

\lud*
\begin{proof}
We first show that unit demand functions are not recoverable. Define a hypothesis class of functions $\mathcal{F}$ which contains $n$ unit demand functions $f_j(\cdot)$ with $f(\{e_1\})  = j / n$ and $f(\{e_i\})  = 1$ for $i \geq 2$, for all integers $1 \leq j \leq n$. We wish to recover function $f_j(\cdot)$ with $j$ picked uniformly at random. With high probability, the sample $\{e_1\}$ is not observed when $k \geq n^{\epsilon}$, so the values of all observed samples are independent of $j$.  Unit demand functions are therefore not recoverable.

Unit demand functions, on the other hand, are $1$-optimizable from samples. With at least $n\log n$ samples, at least one sample contains, with high probability, the best element $e^{\star} := \argmax_{e_i} f(\{e_i\}) $. Any set containing the best element is an optimal solution. Therefore, an algorithm which returns the sample with the largest value obtains an optimal solution with high probability. 
\end{proof}

\newpage

\section{Learning Models}
\label{s:learningmodels}
As a model for statistical learnability we use the notion of \texttt{PAC} learnability due to Valiant~\cite{Valiant84-PAC} and its generalization to real-valued functions \texttt{PMAC} learnability, due to Balcan and Harvey~\cite{BH11-PMAC}.  
Let $\mathcal{F}$ be a hypothesis class of functions $\{f_{1},f_{2},\ldots\}$ where $f_i : 2^N \rightarrow \mathbb{R}$.  
Given precision parameters $\epsilon>0$ and $\delta>0$, the input to a learning algorithm is samples $\{S_{i},f(S_i)\}_{i=1}^{m}$ where the $S_i$'s are drawn i.i.d. from from some distribution $\mathcal{D}$, and the number of samples $m$ is polynomial in $1/\epsilon,1/\delta$ and $n$. The learning algorithm outputs a function $\widetilde{f}: 2^N \rightarrow \mathbb{R}$ that should approximate $f$ in the following sense.
\begin{itemize}
\item $\mathcal{F}$ is \texttt{PAC}-learnable on distribution $\mathcal{D}$ if there exists a (not necessarily polynomial time) learning algorithm  such that for every $\epsilon, \delta >0$:
$$	\Pr_{S_1,\dots,S_m \sim \mathcal{D}} \left[
 		\Pr_{S \sim \mathcal{D}} \left [ \widetilde{f}(S) \neq f(S) \right ] \geq 1-\epsilon
 		\right] \geq 1-\delta$$
\item $\mathcal{F}$ is $\alpha$-\texttt{PMAC}-learnable on distribution $\mathcal{D}$ if there exists a (not necessarily polynomial time) learning algorithm such that for every $\epsilon, \delta >0$:
$$	\Pr_{S_1,\dots,S_m \sim \mathcal{D}} \left[
 		\Pr_{S \sim \mathcal{D}} \Big [  
 			\alpha \cdot \widetilde{f}(S) \leq f(S) \leq   \widetilde{f}(S) 
 		\Big ] \geq 1-\epsilon
 	\right] \geq 1-\delta$$
\end{itemize}
A class $\mathcal{F}$ is \texttt{PAC} (or $\alpha$-\texttt{PMAC}) learnable if it is \texttt{PAC}- ($\alpha$-\texttt{PMAC})-learnable on every distribution $\mathcal{D}$.

\newpage
\section{Discussion}
\label{s:discussion}

\paragraph{Beyond set functions.}   Thinking about models as set functions is a useful abstraction, but optimization from samples can be considered for general optimization problems.  Instead of the max-$k$-cover problem, one may ask whether samples of spanning trees can be used for finding an approximately minimum spanning tree.  Similarly, one may ask whether shortest paths, matching, maximal likelihood in phylogenetic trees, or any other problem where crucial aspects of the objective functions are learned from data, is optimizable from samples.

\paragraph{Coverage functions.} In addition to their stronger learning guarantees, coverage functions have additional guarantees that distinguish them from general monotone submodular functions.  

\begin{itemize}
\item Any polynomial-sized coverage function can be exactly \emph{recovered}, i.e., learned exactly for \emph{all} sets, using polynomially many (adaptive) queries to a value oracle~\cite{chakrabarty2012testing}.  In contrast, there are monotone submodular functions for which no algorithm can recover the function using fewer than exponentially many value queries~\cite{chakrabarty2012testing}.  It is thus interesting that despite being a distinguished class within submodular functions with enough structure to be exactly recovered via adaptive queries, polynomial-sized coverage functions are inapproximable from samples.  
\item In mechanism design, one seeks to design polynomial-time mechanisms which have desirable properties in equilibrium (e.g. truthful-in-expectation).  Although there is an impossibility result for general submodular functions \cite{dughmi2015limitations}, one can show that for coverage functions there is a mechanism which is truthful-in-expectation  \cite{dughmi2011truthful, dughmi2011convex}.
\end{itemize}

\newpage

\section{Hardness of Submodular Functions}
\label{s:lbsubmodular}

Using the hardness framework from Section~\ref{s:framework}, it is relatively easy to show that submodular functions are not $n^{-1/4 + \epsilon}$-\texttt{OPS} over any distribution $\mathcal{D}$.   The good, bad, and masking functions $g,b,m,m^+$ we use are:
\begin{align*}
 g(S) &= |S|,& \\
 b(S) &= \min(|S|, \log n),& \\
  m(S)& =  \min(1,  |S| / n^{1/2}),& \\
 m^+(S) &= n^{-1/4} \cdot \min(n^{1/2}, |S|).& 
\end{align*}
It is easy to show that $\mathcal{F}(g,b,m,m^+)$ is a class of monotone submodular functions (Lemma~\ref{l:monsub}).
To derive the optimal $n^{-1/4+\epsilon}$ impossibility we consider the cardinality constraint $k = n^{1/4 - \epsilon/2}$ and the size of the partition to be $r = n^{1/4}$. We show that $\mathcal{F}(g,b,m,m^+)$ has an $(n^{1/4 - \epsilon}, 0)$-gap. 
\begin{restatable}{rLem}{lgapsub}
\label{l:gapsub}
The class $\mathcal{F}(g,b,m,m^+)$ as defined above has an $(n^{1/4 - \epsilon}, 0)$-gap with $t = n^{1/2+\epsilon/4}$.
\end{restatable}
\begin{proof}
We show that these functions satisfy the properties to have an $(n^{1/4 - \epsilon}, 0)$-gap.
\begin{itemize}
\item \textbf{Identical on small samples.} Assume $|S| \leq n^{1/2 + \epsilon/4}$. Then $|T_{-i}|\cdot |S| / n \leq n^{1/2- \epsilon/2} \cdot n^{1/2 + \epsilon/4} / n \leq n^{-\epsilon/4}$, so by Lemma~\ref{lem:skewed-concentration}, $|S \cap T_{-i}| \leq \log n$ w.p.  $1- \omega(1)$ over $P \sim \UC(\PC)$. Thus $$g(S \cap T_{i}) + b(S \cap T_{- i}) = |S \cap (\cup_{j=1}^r T_j)| $$ with probability $1- \omega(1)$  over $P$.
\item \textbf{Identical on large samples.} Assume $|S| \geq n^{1/2 +\epsilon/4}$. Then $|S \cap M|   \geq n^{1/2}$ with exponentially high probability over $P \sim \UC(\PC)$ by Chernoff bound (Lemma~\ref{l:chernoff}), and $m(S \cap M) = 1$ w.p. at least $1- \omega(1)$.
\item \textbf{Gap $n^{1/4 - \epsilon}$.} Note that $g(S) = k = n^{1/4 - \epsilon/2}$, $b(S) = \log n$, $m^+(S) = n^{-\epsilon/2}$ for $|S| = k$, so $g(S) \geq n^{1/4 - \epsilon} b(S)$ for $n$ large enough and $g(S) = n^{1/4} m^+(S)$.
\item{Curvature $\beta = 0$.} The curvature $\beta = 0$  follows from $g$ being linear.
\end{itemize}
\end{proof}

We show that that we obtain monotone submodular functions.
\begin{lemma}
\label{l:monsub}
The class of functions $\mathcal{F}(g,b,m,m^+)$ is a class of monotone submodular functions.
\end{lemma}
\begin{proof}
We show that the marginal contributions $f_S(e)$ of an element $e \in N$ to a set $S \subseteq N$ are such that $f_S(e) \geq f_T(e) $ for $S \subseteq T$  (submodular) and $f_S(e) \geq 0$ for all $S$ (monotone) for all elements $e$. For $e \in T_j$, for all $j$, this follows immediately from $g$ and $b$ being monotone submodular. For $e \in M$, note that $$f_S(e) = \begin{cases}  - \frac{1}{n^{1/2}}(|S \cap T_i| +\min(|S \cap T_{-i}|, \log n)) + n^{-1/4} & \text{ if } |S \cap M| < n^{1/2} \\ 0 & \text{ otherwise } \end{cases}$$ Since $|S \cap T_i| +\min(|S \cap T_{-i}|, \log n) \leq n^{1/4}$, $f_S(e) \geq f_T(e) $ for $S \subseteq T$ and $f_S(e) \geq 0$ for all $S$.
\end{proof}

  Together with Theorem~\ref{t:gap}, these two lemmas imply the hardness result.

\begin{restatable}{rThm}{tsubmodular}
\label{t:submodular}
For every constant $\epsilon > 0$,  monotone submodular functions are not $n^{-1/4 + \epsilon}$-\texttt{OPS} over any distribution $\mathcal{D}$.
\end{restatable}

\end{document}